\documentclass[onefignum,onetabnum]{siamart171218}

\makeatletter
\AtBeginDocument{%
  \def\refstepcounter@optarg[#1]#2{%
    \cref@old@refstepcounter{#2}%
    \cref@constructprefix{#2}{\cref@result}%
    \@ifundefined{cref@#1@alias}%
      {\def\@tempa{#1}}%
      {\def\@tempa{\csname cref@#1@alias\endcsname}}%
    \protected@edef\cref@currentlabel{%
      [\@tempa][\arabic{#2}][\cref@result]%
      \csname p@#2\endcsname\csname the#2\endcsname}}}
\makeatother

\usepackage{amsfonts,amssymb}
\usepackage{url}
\usepackage{yhmath}

\newsiamthm{assumption}{Assumption}
\newsiamthm{condition}{Condition}
\newsiamremark{remark}{Remark}

\makeatletter
\@removefromreset{theorem}{section}
\@for\@tempa:={lemma,corollary,proposition,definition,assumption,condition,remark}\do{%
  \expandafter\let\csname\@tempa\endcsname\relax
  \expandafter\let\csname end\@tempa\endcsname\relax
  \expandafter\let\csname c@\@tempa\endcsname\relax
  \expandafter\let\csname the\@tempa\endcsname\relax
  \expandafter\let\csname\@tempa*\endcsname\relax
  \expandafter\let\csname end\@tempa*\endcsname\relax}
\makeatother

\theoremstyle{plain}\theoremheaderfont{\normalfont\sc}\theorembodyfont{\normalfont\itshape}
\theoremseparator{.}\theoremsymbol{}
\newtheorem{lemma}{Lemma}
\newtheorem{corollary}{Corollary}
\newtheorem{proposition}{Proposition}
\newtheorem{definition}{Definition}
\newtheorem{assumption}{Assumption}

\theoremheaderfont{\normalfont\itshape}\theorembodyfont{\normalfont}
\newtheorem{remark}{Remark}

\renewcommand{\Re}{\mathrm{Re}}
\newcommand{\E}{\mathbb{E}}
\newcommand{\cA}{\mathcal{A}}
\newcommand{\cS}{\mathcal{S}}
\newcommand{\cB}{\mathcal{B}}
\newcommand{\cH}{\mathcal{H}}
\newcommand{\cT}{\mathcal{T}}
\newcommand{\dd}{\,d}
\newcommand{\Rea}{\operatorname{Re}}
\newcommand{\dt}[1]{\dot{\wideparen{#1}}}

\headers{Stabilization of Yih's Viscosity Jump Interface}{M. Krstic}

\title{Stabilization of Yih's Viscosity Jump Interface\\ in a Channel\thanks{\funding{This work was funded by AFOSR grant FA9550-23-1-0535 and NSF grant ECCS-2151525.}}}

\author{Miroslav Krstic\thanks{Department of Mechanical and Aerospace Engineering, University of California, San Diego, La Jolla, CA 92093-0411, USA (\email{mkrstic@ucsd.edu}).}}

\allowdisplaybreaks
\begin{document}
\maketitle

\begin{abstract}
Two immiscible viscous fluids in pressure-driven channel flow can be unstable through their interface when their viscosities differ, at arbitrarily small Reynolds number --- the instability Yih found in 1967. The unstable state is the interface itself: a curve inside the domain, separating the two fluids, whose displacement is governed by a partial differential equation of its own and which moves the fluid domain with it. Actuation is at one wall only, so the state to be controlled is reached across a fluid. We stabilize this interface at any rate below the limit set by the streamwise mean flow, by static feedback of the two velocity components of that wall, through a Fredholm backstepping transformation acting jointly on the two layers. The target is the two-fluid Stokes problem, shifted, with the transmission of stress between the fluids kept and the coupling of the interface into them removed. Two findings about the two-fluid spectrum carry the design: joined at the interface, the eigenvalue sequences of the two layers, which collide for a dense set of layer ratios when the layers are uncoupled, avoid each other uniformly, so no degeneracy has to be assigned; and the slowest mode left to viscosity is not a bulk mode but the capillary relaxation of the interface, whose rate is linear rather than quadratic in the wavenumber, which sets the band of actuated wavenumbers. The controller is tested on the nonlinear two-fluid channel, switched on against a saturated interfacial wave.
\end{abstract}



\section{Introduction}

\subsection{Problem}

Two incompressible fluids of different viscosity, stacked in a plane channel and driven by a pressure gradient, flow with a piecewise parabolic profile whose slope jumps at the interface. Yih \cite{Yih67} showed that this viscosity jump can destabilize long interfacial waves at arbitrarily small Reynolds number, depending on the viscosity ratio and the layer thicknesses. The mode is not a mode of either fluid. It is carried by the displacement of the interface, which is an unknown of the problem in its own right, and it feeds back into the fluids through a single term in the tangential-velocity condition at the interface. The instability underlies the breakdown of lubricated transport of viscous oil by a water layer, the wave regimes of stratified oil--water pipelines, and the interfacial defects of coextruded films and coating flows --- configurations in which the shape of the interface is the product.

\subsection{Literature}

The interface is a moving boundary. Backstepping feedback stabilization of moving-boundary PDE systems has been developed over the past decade for a thermal phase-change front (the Stefan problem) \cite{KDK19}, a piston between two columns of compressible gas \cite{KK22}, the growth cone of a neuron \cite{DKK24}, and the shock front of a traffic wave \cite{YDZK21}. A string whose endpoint moves is controlled from the boundary in \cite{Gug07,Gug08}, and a point mass moving in a viscous fluid is steered by a force on the mass in \cite{CMRT15}. In each of these the domain is one-dimensional and the moving boundary is a point, whose position is governed by an ordinary differential equation coupled to the PDE through its boundary values. A subsystem joined to others at both of its ends is a sandwich system: in \cite{WK20} an ODE sandwiched between two PDEs, the coupled hyperbolic PDEs of a cable on one side and the transport PDE of a sensor delay on the other, and in \cite{WK21} a PDE sandwiched between two ODEs, with the PDE domain varying in time along a prescribed law. The plant here has that structure with each part one dimension higher: two-dimensional fluids on either side and a one-dimensional interface between them. The domain variation is not prescribed either, since the interface is the state that moves the boundaries of the two fluid domains. The present problem is the next case in that sequence: the domain is two-dimensional, the moving boundary is a curve, and its position is governed by a partial differential equation of its own --- the kinematic condition --- coupled to the fluids along its whole length.

Outside of backstepping, Munteanu \cite{Mun21} posed and solved the boundary stabilization problem for two viscous incompressible fluids in a channel: the Navier--Stokes equations coupled to a convective Cahn--Hilliard equation for the relative density of one species, with the Poiseuille profile as the target flow, stabilized by feedback of the tangential velocity at both walls, the feedback given in explicit form. There the two fluids are described by an order parameter and the interface between them is a diffuse layer of positive thickness, so the fluid domain is fixed and its boundary is the two walls. The present paper treats the same two fluids with the interface sharp, as a moving boundary carrying the state.

Controllability of fluid--structure interactions, in which an elastic structure moves inside a fluid, has been studied \cite{IT07,BO08}, and feedback stabilization of a fluid--beam system by a Riccati-based design in \cite{Ray10}. An interface between two fluids has no structure of its own: it carries no stiffness and no mass, only capillarity, and it moves because the fluids move.

\subsection{This Paper}

Two features set the problem tackled in this paper apart from those it descends from. Unlike single-fluid Poiseuille flow \cite{VK07,Kr26}, which is unstable only above a Reynolds number in the thousands, the two-fluid channel can be unstable at arbitrarily small Reynolds number, so the instability is not a high-Reynolds-number effect to be suppressed but a low-Reynolds-number one that the equilibrium profile carries into every regime. And unlike the moving-boundary systems in \cite{KDK19,KK22,DKK24,YDZK21}, where the boundary follows the interior fields and the interior is where any instability would arise, here the instability originates at the boundary and is manifested there: the unstable mode is the interface displacement itself, and the fluids serve only to transmit it.

This paper controls that instability from the wall, with the two velocity components of the upper wall, by the Fredholm backstepping design developed for single-fluid Poiseuille flow in \cite{Kr26}, which carries the rapid stabilization of \cite{CL14,CL15} to that plant. Everything the design uses is taken from there: the Stokes problem as target, the finite-rank pre-feedback through the normal wall velocity that makes every mode controllable from the tangential one, the Riesz basis of the plant obtained by perturbation from the target's, and the Fredholm transformation built on that basis. None of it is re-derived here; where a proof carries over it is cited, and where the second fluid changes it --- the spectrum, the wall coefficients, the slow mode, the mean flow --- it is proved. What is added is the object: an interface between two fluids, an unstable state attached to a curve inside the domain and governed by its own equation, which \cite{Kr26} does not contain.

The lower wall is unactuated. So the interface, the object to be controlled, sits at an interior point of the channel and is reached only across the actuated fluid.

Per wavenumber, the linearized plant is a pair of fourth-order Orr--Sommerfeld operators on the two layers, joined at the interface by four transmission conditions, with the interface displacement as an additional scalar state attached there. One of the transmission conditions carries a time derivative when the densities differ. The unstable coupling of the interface into the fluids is of rank one.

\subsection{Contributions}

\begin{itemize}
\item \emph{Rapid stabilization of a PDE-governed moving boundary.} The Yih interfacial instability of two fluids across a viscosity jump in a channel is stabilized at every prescribed rate below the mean-flow controllability limit, with wall actuation in one of the two fluids, the other wall being unactuated. Where the two fluids in a channel have been stabilized with the interface diffuse \cite{Mun21}, on a fixed domain and at an unprescribed rate, here the interface is sharp: the state includes its displacement, the fluid domain moves with it, and the moving boundary is governed by a partial differential equation of its own in one space dimension rather than by an ordinary differential equation at a point, as in \cite{KDK19,KK22,DKK24,YDZK21}. Every mode that deforms the interface admits an arbitrary rate; what limits the achievable rate is the streamwise mean flow alone, which one wall velocity can fail to reach, and that limit is shown to be attainable and to tighten with the Reynolds number. This is possible because the interface, although an interior object reached only across the actuated fluid, is a single scalar per wavenumber whose only destabilizing coupling to the fluids is of rank one, and a Fredholm transformation acting jointly on the two layers removes it while leaving the transmission conditions between the layers intact.
\end{itemize}

Three mathematical obstacles stood in the way, met by the following.

\begin{itemize}
\item \emph{Coupling two layers simplifies their spectrum.} The interface does not make the spectrum harder; it makes it simpler. Two viscous layers considered separately have eigenvalue sequences that come arbitrarily close for a dense set of viscosity and thickness ratios, and a design would then have to assign near-degenerate pairs. Joined at the interface, the sequences avoid each other uniformly: at every admissible parameter value the high modes of the coupled problem are simple and uniformly separated, and the one-fluid Fredholm construction applies to the two-fluid channel without a new spectral device.
\item \emph{What an interfacial target may delete.} The target of the design is the two-fluid Stokes problem itself, shifted and with its finitely many low eigenvalues separated by a finite-rank term, with the viscosity jump kept in its transmission conditions. Deleting that jump makes the target dissipative for a reason unrelated to the physics; keeping it makes the dissipation the strain-rate integral of the two fluids. This identifies what a backstepping target for an interfacial problem may delete --- the coupling of the interface displacement into the fluids --- and what it may not --- the transmission of stress between them.
\item \emph{Capillary relaxation sets the actuated band.} What a wall feedback acting on finitely many wavenumbers must outrun is not the viscous decay of the fluids but the capillary relaxation of the interface, which is slower: it is linear in the wavenumber where the bulk modes decay quadratically. Computing that rate with its constant, from the mobility of a viscous interface, makes the controlled band explicit and linear in the requested rate.
\end{itemize}

The main theorem is stated in Section \ref{sec:main}; the design and its analysis follow, in the order of the dependencies among them.

\section{Two-Fluid Channel and Its Linearization}\label{sec:model}

Fluid 1 occupies $0<y<d+\eta(t,x)$ and fluid 2 occupies $d+\eta(t,x)<y<1$, both $L$-periodic in $x$, with densities $\rho_j$, viscosities $\mu_j$, and the interface at $y=d+\eta(t,x)$. Lengths are scaled by the channel height $H$, velocities by the centerline speed $U^e$ of the single-fluid Poiseuille profile with the same pressure gradient, pressure by $\rho_1(U^e)^2$, and $\rho_1=\mu_1=1$ in nondimensional terms. With
\begin{equation}\label{eq:groups}
\Re=\frac{\rho_1U^eH}{\mu_1},\quad M=\frac{\mu_2}{\mu_1},\quad \Sigma=\frac{\sigma}{\rho_1(U^e)^2H},\quad G=\frac{gH}{(U^e)^2},
\end{equation}
the fluids obey, in their respective domains,
\begin{subequations}\label{eq:ns}
\begin{align}
\partial_xu_j+\partial_yv_j &= 0, \label{eq:cont}\\
\rho_j\big(\partial_tu_j+u_j\partial_xu_j+v_j\partial_yu_j\big) &= -\partial_xp_j+\frac{\mu_j}{\Re}\big(\partial_x^2+\partial_y^2\big)u_j, \label{eq:mx}\\
\rho_j\big(\partial_tv_j+u_j\partial_xv_j+v_j\partial_yv_j\big) &= -\partial_yp_j+\frac{\mu_j}{\Re}\big(\partial_x^2+\partial_y^2\big)v_j-\rho_jG, \label{eq:my}
\end{align}
\end{subequations}
with no slip at the unactuated wall and the two actuated components at the upper wall,
\begin{subequations}\label{eq:wallbc}
\begin{align}
u_1(t,x,0)&=v_1(t,x,0)=0, \label{eq:walls}\\
u_2(t,x,1)&=U_c(t,x),\qquad v_2(t,x,1)=V_c(t,x), \label{eq:walls2}
\end{align}
\end{subequations}
and at the interface $y=d+\eta(t,x)$,
\begin{subequations}\label{eq:ifc}
\begin{align}
\eta_t+u_j\eta_x&=v_j,\qquad [u_1-u_2]=[v_1-v_2]=0,\qquad (T_2-T_1)n=\Sigma\kappa n, \label{eq:interface}\\
T_j&=-p_jI+\frac{\mu_j}{\Re}\big(\nabla\mathbf{u}_j+\nabla\mathbf{u}_j^{\mathrm T}\big), \label{eq:stress}\\
n&=\frac{(-\eta_x,1)}{\sqrt{1+\eta_x^2}},\qquad \kappa=-\frac{\eta_{xx}}{(1+\eta_x^2)^{3/2}} . \label{eq:curv}
\end{align}
\end{subequations}
Here \eqref{eq:stress}--\eqref{eq:curv} define the stress and the normal and curvature of the graph $y=d+\eta$. The mean flux is held at its equilibrium value.

The flat interface $\eta\equiv0$ with $v_j\equiv0$, $u_j=U_j(y)$, and $p_j=-Px+p_I-\rho_jG(y-d)$, normalized by $P\Re=8$, gives
\begin{subequations}\label{eq:base}
\begin{align}
U_1(y)&=4y(G_0-y),\qquad U_2(y)=\frac{4}{M}(1-y)(1+y-G_0), \label{eq:profile}\\
G_0&=\frac{1-(1-M)d^2}{1-(1-M)d}, \label{eq:G0}
\end{align}
\end{subequations}
continuous at $y=d$ with $\mu_1U_1'(d)=\mu_2U_2'(d)$, so that the slope jumps by the factor $1/M$, and with $\mu_1U_1''=\mu_2U_2''=-8$. The interface speed and slope are
\begin{equation}\label{eq:UI}
U_I=U_1(d)=\frac{4Md(1-d)}{1-(1-M)d},\qquad U_1'(d)=\frac{4\big[1-2d+(1-M)d^2\big]}{1-(1-M)d}.
\end{equation}

Integrals over a layer are written out, $\int_0^d$ for fluid 1 and $\int_d^1$ for fluid 2, and $\sum_j\int_j$ abbreviates their sum. Linearizing about \eqref{eq:profile}, transferring the interface conditions to $y=d$, expanding in $e^{2\pi ikx}$ with $k=m/L$, and eliminating the streamwise velocity and the pressure through $u_j=\frac{i}{2\pi k}\partial v_j$ and the streamwise momentum equation, as in \cite{Kr26}, gives the plant per wavenumber in the normal velocities $v_j(y,t)$ and the interface displacement $\eta(t)$. With $\E:=\partial^2-4\pi^2k^2$, $\partial:=\partial_y$,
\begin{subequations}\label{eq:plant}
\begin{align}
\rho_j\partial_t\E v_j &= \frac{\mu_j}{\Re}\E^2v_j-2\pi ik\rho_jU_j\E v_j+2\pi ik\rho_jU_j''v_j,\quad j=1,2, \label{eq:p1}\\
0 &= v_1(0,t)=\partial v_1(0,t), \label{eq:p3}\\
V_c(k,t) &= v_2(1,t), \label{eq:p4}\\
-2\pi ik\,U_c(k,t)&=\partial v_2(1,t), \label{eq:p4b}\\
0 &= \big[v_1-v_2\big]_{y=d}, \label{eq:p5}\\
0 &= \partial\big[v_1-v_2\big]_{y=d}-2\pi ik\Big(1-\frac1M\Big)U_1'(d)\,\eta, \label{eq:p6}\\
\big[\mu_1\E v_1-\mu_2\E v_2\big]_{y=d} &= -8\pi^2k^2(\mu_1-\mu_2)\,v_2(d,t), \label{eq:p7}\\
\big[\mu_1\partial\E v_1-\mu_2\partial\E v_2\big]_{y=d} &= \Re(\rho_1-\rho_2)\Big(\dt{\partial v_2}(d,t)+2\pi ikU_I\,\partial v_2(d,t)\Big) \nonumber\\
&\quad+8\pi^2k^2\big[\mu_1\partial v_1-\mu_2\partial v_2\big]_{y=d} \nonumber\\
&\quad+2\pi ik\,\Re\,\frac{\rho_2-\rho_1}{M}U_1'(d)\,v_2(d,t) \nonumber\\
&\quad-4\pi^2k^2\Re\big[(\rho_2-\rho_1)G-4\pi^2k^2\Sigma\big]\eta, \label{eq:p8}\\
\dot\eta &= -2\pi ikU_I\,\eta+v_2(d,t). \label{eq:p9}
\end{align}
\end{subequations}
Equation \eqref{eq:p6} carries the Yih coupling: the interface displacement enters the slope condition with the coefficient $(1-1/M)U_1'(d)$, which vanishes for equal viscosities and at $d=1/(1+\sqrt M)$, the interface at the maximum of the profile. Equation \eqref{eq:p8} is the normal-stress balance with the pressure eliminated; its time derivative of a trace makes the per-wavenumber system a pencil when $\rho_1\ne\rho_2$. Equation \eqref{eq:p9} is the kinematic condition. At $k=0$, $v\equiv0$, the mean of $\eta$ is conserved, and the streamwise mean obeys a two-layer heat equation with transmission, treated in Section \ref{sec:phys}.

\section{Controller and Main Result}\label{sec:main}

\subsection{Controller and stabilization theorem}

\begin{definition}[Energy spaces]\label{def:H}
For $k\ne0$, $\cH'_k$ is the space of triples $x=(v_1,v_2,\eta)$ with $v_1\in H^1(0,d)$, $v_2\in H^1(d,1)$, $v_1(0)=0$, $v_1(d)=v_2(d)$, and $\cH_k\subset\cH'_k$ is the subspace with $v_2(1)=0$, both with the inner product
\begin{equation}\label{eq:ip}
\langle x,x'\rangle_k=\sum_{j=1}^2\rho_j\Big(\langle\partial v_j,\partial v'_j\rangle+4\pi^2k^2\langle v_j,v'_j\rangle\Big)+4\pi^2k^2\varpi^+_k\,\eta\bar\eta',
\end{equation}
with $\varpi^+_k=1+|\rho_1-\rho_2|G+4\pi^2k^2\Sigma$. The physical state space $X$ consists of the fields $(u_1,u_2,v_1,v_2,\eta)$ on the channel with
\begin{equation}\label{eq:X}
\|x\|_X^2=\sum_{j=1}^2\rho_j\big\|(u_j,v_j)\big\|^2+\sum_k\varpi^+_k|\eta(k)|^2<\infty,\quad \int_{\mathbb{T}_L}\eta\dd x=0,
\end{equation}
\end{definition}

All $x$-integrals are taken with the normalized measure $\dd x/L$ on $\mathbb{T}_L=(0,L)$, so Parseval reads $\|\phi\|^2=\sum_m|\phi_m|^2$. For $k\ne0$ the modal field energy equals the kinetic energy of the velocity pair, $\rho_j(\|u_j\|^2+\|v_j\|^2)=\frac{\rho_j}{4\pi^2k^2}(\|\partial v_j\|^2+4\pi^2k^2\|v_j\|^2)$, and the interface weight is the potential energy of gravity and capillarity where that is positive, made positive throughout by the constant.

\begin{assumption}[Surface tension]\label{ass:sigma}
$\Sigma>0$.
\end{assumption}

\begin{definition}[Mean-flow rate limit]\label{def:qmean}
\begin{equation}\label{eq:qmean}
q_{\rm mean}:=\inf\Big\{-\sigma:\ \sigma\in\sigma(\cA_0),\ \ker(\cA_0-\sigma)\cap\ker B_0^*\ne\{0\}\Big\},
\end{equation}
$q_{\rm mean}:=+\infty$ if that set is empty, where $\cA_0$ is the operator of the streamwise mean, $u\mapsto\big(-P'+\frac{\mu_j}{\Re}\partial^2u_j\big)/\rho_j$ on $\{u:\ u_1\in H^2(0,d),\ u_2\in H^2(d,1),\ \int_0^1u\dd y=0\}$ with $u_1(0)=u_2(1)=0$, $[u]_{y=d}=[\mu_1\partial u_1-\mu_2\partial u_2]_{y=d}=0$ and $P'$ the multiplier of the zero-mean constraint, and $B_0^*w=-\frac{\mu_2}{\Re}\overline{\partial w_2(1)}$ is the input functional of $U_c(0,\cdot)$.
\end{definition}

The set in \eqref{eq:qmean} collects the eigenvalues of the streamwise mean that the Hautus test declares uncontrollable from $U_c(0,\cdot)$, so $-q_{\rm mean}$ is the slowest of them and the wall input cannot move it. Rates below $q_{\rm mean}$ are assigned by Theorem \ref{thm:main}; rates above it are not available at $k=0$.

\begin{assumption}[Interface position]\label{ass:d}
The interface lies no closer to the actuated wall than the fraction $\sqrt{\mu_2}/(\sqrt{\mu_1}+\sqrt{\mu_2})$ of the channel height:
\begin{equation}\label{eq:dbound}
d\ \le\ \frac{1}{1+\sqrt{\mu_2/\mu_1}} .
\end{equation}
\end{assumption}

\begin{remark}\label{rem:simple}
Assumption \ref{ass:d} is imposed for simplicity of interpretation and exposition. It is sufficient, not necessary: what the design needs at $k\ne0$ is weaker, and is stated and discussed in Remark \ref{rem:weak}.
\end{remark}

The design maps the closed-loop plant onto a target by the Fredholm transformation $(w_1,w_2,\xi)=\cT(v_1,v_2,\eta)$,
\begin{subequations}\label{eq:fred}
\begin{align}
w_j(y,t)&=v_j(y,t)-\int_0^{d}\!f_{j1}(y,\zeta)v_1(\zeta,t)\dd\zeta-\int_d^{1}\!f_{j2}(y,\zeta)v_2(\zeta,t)\dd\zeta-g_j(y)\eta(t), \label{eq:fredw}\\
\xi(t)&=\eta(t)-\int_0^{d}\!h_1(\zeta)v_1(\zeta,t)\dd\zeta-\int_d^{1}\!h_2(\zeta)v_2(\zeta,t)\dd\zeta, \label{eq:fredxi}
\end{align}
\end{subequations}
for $j=1,2$, with kernels $f_{ji}$ and profiles $g_j,h_i$ to be found, so that $(w_1,w_2,\xi)$ obeys the target system \eqref{eq:target}: two Stokes fluids with their stress transmission intact, the interface decoupled from them, and every block shifted by $q$, which decays at rate $q$ by Proposition \ref{prop:target}. The integrals run over both layers in full, so \eqref{eq:fred} is Fredholm and not Volterra, and the kernels are obtained in Section \ref{sec:design} from a basis rather than from a kernel equation.

\begin{theorem}[Two wall inputs stabilize the two-fluid channel at any rate below the mean-flow limit]\label{thm:main}
Let Assumptions \ref{ass:sigma} and \ref{ass:d} hold, and let $0<q<q_{\rm mean}$. There exist $K>0$ and, for each $0<|k|<K$, a bounded linear feedback $F_k=(K_V,F^U_k):\cH'_k\to\mathbb{C}^2$, $K_V$ of finite rank, with $F_{-k}\bar x=\overline{F_kx}$, and a bounded feedback $F_0$ on the streamwise mean, such that under
\begin{subequations}\label{eq:fb}
\begin{align}
\big(V_c(k,t),U_c(k,t)\big)&=F_kx(k,t)\qquad(0<|k|<K), \label{eq:feedback}\\
U_c(0,t)&=F_0u(0,\cdot,t),\qquad V_c=U_c=0\quad(|k|\ge K), \label{eq:feedback2}
\end{align}
\end{subequations}
where $F_k$ acts as
\begin{subequations}\label{eq:physfb}
\begin{align}
V_c(k,t)&=\int_0^{d}\!a_1(\zeta)v_1(\zeta,t)\dd\zeta+\int_d^{1}\!a_2(\zeta)v_2(\zeta,t)\dd\zeta+a_0\,\eta(t), \label{eq:physV}\\
U_c(k,t)&=\frac{i}{2\pi k}\Big(\int_0^{d}\!\partial_yf_{21}(1,\zeta)v_1(\zeta,t)\dd\zeta+\int_d^{1}\!\partial_yf_{22}(1,\zeta)v_2(\zeta,t)\dd\zeta+g_2'(1)\eta(t)\Big), \label{eq:physU}
\end{align}
\end{subequations}
with $a_1,a_2,a_0$ the finite-rank profiles of the pre-feedback $K_V$ of Lemma \ref{lem:heymann}, which is fixed before \eqref{eq:fred} is built, and the gains of $U_c$ the wall traces of the kernels of \eqref{eq:fred}, obtained by imposing $\partial w_2(1)=0$ on \eqref{eq:fredw} and using \eqref{eq:p4b}, the closed loop consisting of the linearization of \eqref{eq:ns}--\eqref{eq:ifc} about \eqref{eq:base}, whose modes are \eqref{eq:plant}, under the feedback \eqref{eq:fb}, has the following two properties.

\begin{enumerate}
\item[(i)] \emph{Well-posedness.} It generates an analytic semigroup on the closed subspace $X_F\subset X$ on which $v_2(k,1)=V_c(k)$ for $0<|k|<K$, so for every initial datum in $X_F$ the solution exists and is unique in $C([0,\infty);X)\cap C^1((0,\infty);X)$, is real for real data, and has pressure $p_j(\cdot,t)\in L^2(\mathbb{T}_L\times(0,d))\ \text{and}\ L^2(\mathbb{T}_L\times(d,1))$ for $t>0$.
\item[(ii)] \emph{Decay at the prescribed rate.}
\begin{equation}\label{eq:decay}
\|x(t)\|_X\le C\,e^{-qt}\,\|x(0)\|_X,\qquad t\ge0,
\end{equation}
\end{enumerate}
for all $\mu_1,\mu_2,\rho_1,\rho_2>0$, $\Re>0$, and $L>0$, where
\begin{equation}\label{eq:K}
K=\max\Big\{K_0,\ \frac{2q\,(\mu_1+\mu_2)}{\pi\Re\,\Sigma}\Big\},
\end{equation}
$K_0$ depends on $\mu_j,\rho_j,d,\Re,G,\Sigma$ and not on $q$, and $C$ depends on those, on $L$, and on $q$.
\end{theorem}

The feedback is real for real data by the conjugate symmetry.

\subsection{Why two wall inputs, and what they achieve}

Controllability from the two wall velocities is decided by the Fattorini--Hautus criterion, applied to the wall functionals and the adjoint at the interface. The plant per wavenumber is $\dot x=\cA_kx+B_{V_c}V_c+B_{U_c}U_c$ on $\cH_k$, with $\cA_k$ the operator of \eqref{eq:plant} at $V_c=U_c=0$ and $B_{V_c},B_{U_c}$ the boundary control operators of the two wall traces in \eqref{eq:p4}. Their adjoint pairings are
\begin{equation}\label{eq:pairings}
\langle B_{V_c},\psi\rangle=-\frac{\mu_2}{\Re}\overline{\partial\E\psi_2(1)},\quad \langle B_{U_c},\psi\rangle=-2\pi ik\frac{\mu_2}{\Re}\overline{\E\psi_2(1)},
\end{equation}
obtained by pairing $\frac{\mu_2}{\Re}\E^2$ with $\psi_2$ and integrating by parts twice, the adjoint domain requiring $\psi_2(1)=\partial\psi_2(1)=0$. Every eigenvalue is controllable from $(V_c,U_c)$ if and only if no eigenvector of $\cA_k^*$ annihilates both, which Theorem \ref{thm:two} establishes but for one exception excluded by Assumption \ref{ass:d}. Its proof uses the interface conditions on such an eigenvector.

\begin{lemma}[Adjoint transmission conditions]\label{lem:adjoint}
Let $w_k:=4\pi^2k^2\varpi^+_k$ be the weight of $\eta$ in \eqref{eq:ip}. Then $\psi=(\zeta_1,\zeta_2,\chi)$ satisfies $\cA_k^*\psi=\bar\lambda\psi$ if and only if
\begin{subequations}\label{eq:adj}
\begin{align}
\frac{\mu_j}{\Re}\E^2\zeta_j+2\pi ik\rho_j\E(U_j\zeta_j) &\nonumber\\
-2\pi ik\rho_jU_j''\zeta_j &= \bar\lambda\rho_j\E\zeta_j, \label{eq:a1}\\
0 &= \zeta_1(0)=\partial\zeta_1(0)=\zeta_2(1)=\partial\zeta_2(1), \label{eq:a2}\\
0 &= \big[\zeta_1-\zeta_2\big]_{y=d}=\partial\big[\zeta_1-\zeta_2\big]_{y=d}, \label{eq:a3}\\
\big[\mu_1\E\zeta_1-\mu_2\E\zeta_2\big]_{y=d} &= -8\pi^2k^2(\mu_1-\mu_2)\,\zeta(d), \label{eq:a4}\\
\big[\mu_1\partial\E\zeta_1-\mu_2\partial\E\zeta_2\big]_{y=d} &= \Re(\rho_1-\rho_2)\big(\bar\lambda-2\pi ikU_I\big)\partial\zeta(d) \nonumber\\
&\quad+8\pi^2k^2(\mu_1-\mu_2)\,\partial\zeta(d) \nonumber\\
&\quad-2\pi ik\,\Re\,\rho_1\Big(1-\frac1M\Big)U_1'(d)\,\zeta(d)-\Re\,w_k\,\chi, \label{eq:a5}\\
w_k\Re\big(\bar\lambda-2\pi ikU_I\big)\chi &= 2\pi ik\Big(1-\frac1M\Big)U_1'(d)\,\mu_1\,\E\zeta_1(d) \nonumber\\
&\quad+16\pi^3ik^3\Big(1-\frac1M\Big)U_1'(d)\,\mu_1\,\zeta(d) \nonumber\\
&\quad-2\pi ik\Big(1-\frac1M\Big)U_1'(d)\,\Re\rho_1\big(\bar\lambda-2\pi ikU_I\big)\,\zeta(d) \nonumber\\
&\quad-4\pi^2k^2\Re\big[(\rho_1-\rho_2)G+4\pi^2k^2\Sigma\big]\zeta(d), \label{eq:a6}
\end{align}
\end{subequations}
where $\zeta(d)$, $\partial\zeta(d)$ denote the common traces of \eqref{eq:a3}.
\end{lemma}

\begin{theorem}[Every eigenvalue is controllable from the two wall inputs]\label{thm:two}
Let $k\ne0$, let $\lambda$ be an eigenvalue of $\cA_k$, and let $\psi=(\zeta_1,\zeta_2,\chi)$ satisfy $\cA_k^*\psi=\bar\lambda\psi$ and $\langle B_{V_c},\psi\rangle=\langle B_{U_c},\psi\rangle=0$. Then $\psi=0$, unless $\lambda=\lambda_0:=-2\pi ikU_I$ and the solution $\varphi$ of
\begin{subequations}\label{eq:exc}
\begin{align}
\frac{\mu_1}{\Re}\E^2\varphi+2\pi ik\rho_1\E(U_1\varphi)-2\pi ik\rho_1U_1''\varphi&=\overline{\lambda_0}\,\rho_1\E\varphi,\qquad 0<y<d, \label{eq:exception}\\
\varphi(d)=\partial\varphi(d)=\E\varphi(d)&=0,\qquad \partial\E\varphi(d)=1, \label{eq:exception2}
\end{align}
\end{subequations}
satisfies $\varphi(0)=\partial\varphi(0)=0$.
\end{theorem}

\begin{proof}
By \eqref{eq:pairings} and \eqref{eq:a2}, $\zeta_2(1)=\partial\zeta_2(1)=\E\zeta_2(1)=\partial\E\zeta_2(1)=0$, so \eqref{eq:a1} on $(d,1)$ gives $\zeta_2\equiv0$, and all traces of $\zeta_2$ at $y=d$ vanish. Then \eqref{eq:a3} gives $\zeta_1(d)=\partial\zeta_1(d)=0$, \eqref{eq:a4} gives $\E\zeta_1(d)=0$, and \eqref{eq:a5} gives $\mu_1\partial\E\zeta_1(d)=-\Re w_k\chi$, so $\zeta_1$ is $-\frac{\Re w_k}{\mu_1}\chi\,\varphi$ with $\varphi$ as in \eqref{eq:exception} at the eigenvalue $\lambda$ in place of $\lambda_0$. Equation \eqref{eq:a6} with $\zeta_1(d)=\E\zeta_1(d)=0$ reads $w_k\Re(\bar\lambda-2\pi ikU_I)\chi=0$. If $\lambda\ne\lambda_0$ then $\chi=0$, hence $\zeta_1\equiv0$ and $\psi=0$. If $\lambda=\lambda_0$, then $\chi\ne0$ is possible only if $\zeta_1=-\frac{\Re w_k}{\mu_1}\chi\varphi$ satisfies \eqref{eq:a2} at $y=0$, which is the stated exception.
\end{proof}

\begin{remark}
When $(1-\tfrac1M)U_1'(d)=0$ and $(\rho_2-\rho_1)G=4\pi^2k^2\Sigma$, the vector $(0,0,\eta)$ is an eigenvector of the open-loop plant at $\lambda_0=-2\pi ikU_I$, with vanishing wall traces. It is nonetheless controllable from the wall, Theorem \ref{thm:two}: the feedback acts on $\eta$, and the wall inputs reach $\eta$ through $v_2(d)$ in \eqref{eq:p9}.
\end{remark}

\begin{proposition}[Exception of Theorem \ref{thm:two} cannot occur under Assumption \ref{ass:d}]\label{prop:phi}
Under Assumption \ref{ass:d}, the solution $\varphi$ of \eqref{eq:exc} does not satisfy $\varphi(0)=\partial\varphi(0)=0$. If Assumption \ref{ass:d} fails and it does, then
\begin{equation}\label{eq:phiRe}
\Re\ \ge\ \frac{\sqrt2\,\pi^2\mu_1}{\rho_1\,d^2\,\|U_1'\|_{L^\infty(0,d)}}.
\end{equation}
\end{proposition}

\begin{proof}
With $\overline{\lambda_0}=2\pi ikU_I$ and $\E(U_1\varphi)-U_1''\varphi=U_1\E\varphi+2U_1'\varphi'$, equation \eqref{eq:exception} reads $\frac{\mu_1}{\Re}\E^2\varphi+2\pi ik\rho_1[(U_1-U_I)\E\varphi+2U_1'\varphi']=0$. Suppose $\varphi(0)=\partial\varphi(0)=0$. Pairing with $\bar\varphi$ on $(0,d)$, all boundary terms vanish, since $\varphi$ and $\varphi'$ vanish at both ends, and $\frac{\mu_1}{\Re}\|\E\varphi\|^2+2\pi ik\rho_1J=0$ with $J:=\int_0^d[(U_1-U_I)\E\varphi+2U_1'\varphi']\bar\varphi$. Integrating by parts, $J=-\int_0^d(U_1-U_I)(|\varphi'|^2+4\pi^2k^2|\varphi|^2)+\int_0^dU_1'\varphi'\bar\varphi$, and $\Rea\int_0^dU_1'\varphi'\bar\varphi=-\frac12\int_0^dU_1''|\varphi|^2=4\int_0^d|\varphi|^2$. The imaginary part of the paired equation is $2\pi k\rho_1\Rea J=0$, hence
\begin{equation}\label{eq:phiid}
\int_0^d(U_1-U_I)\big(|\varphi'|^2+4\pi^2k^2|\varphi|^2\big)dy=4\int_0^d|\varphi|^2dy>0 .
\end{equation}
Pairing instead with $(U_1-U_I)\bar\varphi$, the transport term is $-\int_0^d(U_1-U_I)^2(|\varphi'|^2+4\pi^2k^2|\varphi|^2)$, real, so the real part of the paired equation is $\Rea\int_0^d\E\varphi\,\overline{\E((U_1-U_I)\varphi)}=0$, and with $\E((U_1-U_I)\varphi)=(U_1-U_I)\E\varphi+2U_1'\varphi'-8\varphi$ this gives
\begin{equation}\label{eq:phiid2}
\int_0^d(U_1-U_I)|\E\varphi|^2dy=-16\int_0^d|\varphi'|^2dy<0 .
\end{equation}
Under \eqref{eq:dbound}, $1-2d+(1-M)d^2\ge0$, so $U_1'\ge0$ on $[0,d]$, $U_1\le U_I$ there, and the left side of \eqref{eq:phiid} is nonpositive: contradiction. In general, the real part of the paired equation gives $\frac{\mu_1}{\Re}\|\E\varphi\|^2=2\pi|k|\rho_1|\mathrm{Im}\,J|\le2\pi|k|\rho_1\|U_1'\|_\infty\|\varphi'\|\|\varphi\|$; with $\|\E\varphi\|^2=\|\varphi''\|^2+8\pi^2k^2\|\varphi'\|^2+16\pi^4k^4\|\varphi\|^2$, the Poincar\'e inequalities $\|\varphi''\|\ge\frac{\pi}{d}\|\varphi'\|\ge\frac{\pi^2}{d^2}\|\varphi\|$ give $\Re\ge\mu_1\pi^3/(\rho_1\|U_1'\|_\infty2\pi|k|d^3)$, and $8\pi^2k^2\|\varphi'\|^2+16\pi^4k^4\|\varphi\|^2\ge2\sqrt2(2\pi|k|)^3\|\varphi'\|\|\varphi\|$ gives $\Re\ge2\sqrt2\mu_1(2\pi k)^2/(\rho_1\|U_1'\|_\infty)$; the larger of the two bounds, minimized over $|k|$, is \eqref{eq:phiRe}.
\end{proof}

\begin{remark}\label{rem:weak}
As announced in Remark \ref{rem:simple}, Assumption \ref{ass:d} and its restriction \eqref{eq:dbound} on the interface position $d$ are imposed because they are simple to state, and are not what the design needs. What it needs at $k\ne0$ is only that the Cauchy solution of \eqref{eq:exc} not vanish to second order at the unactuated wall, for which \eqref{eq:dbound} is sufficient by Proposition \ref{prop:phi}. Above the bound no failure has been found: none in the Stokes limit, none below \eqref{eq:phiRe}, none in searches over $(d,\Re,k)$, and none can satisfy the two localization identities \eqref{eq:phiid} and \eqref{eq:phiid2} at once. A proof that none exists is not available here.
\end{remark}

\section{Target and Its Stability}\label{sec:target}

The target of the Fredholm transformation \eqref{eq:fred} is the plant with the transport and reaction terms scaled out, the two $U_1'(d)$ terms and the gravity--capillarity term deleted, the kinematic forcing removed, and every block shifted by $q$:
\begin{subequations}\label{eq:target}
\begin{align}
\rho_j\partial_t\E w_j &= \frac{\mu_j}{\Re}\E^2w_j-q\rho_j\E w_j,\qquad j=1,2, \label{eq:t1}\\
0 &= w_1(0)=\partial w_1(0)=w_2(1)=\partial w_2(1), \label{eq:t2}\\
0 &= \big[w_1-w_2\big]_{y=d}=\partial\big[w_1-w_2\big]_{y=d}, \label{eq:t3}\\
\big[\mu_1\E w_1-\mu_2\E w_2\big]_{y=d} &= -8\pi^2k^2(\mu_1-\mu_2)\,w_2(d,t), \label{eq:t4}\\
\big[\mu_1\partial\E w_1-\mu_2\partial\E w_2\big]_{y=d} &= \Re(\rho_1-\rho_2)\Big(\dt{\partial w_2}(d,t)+q\,\partial w_2(d,t)\Big) \nonumber\\
&\quad+8\pi^2k^2\big[\mu_1\partial w_1-\mu_2\partial w_2\big]_{y=d}, \label{eq:t5}\\
\dot\xi &= -\big(2\pi ikU_I+q\big)\xi . \label{eq:t6}
\end{align}
\end{subequations}
Equations \eqref{eq:t4}--\eqref{eq:t5} are the exact transmission conditions of two Stokes fluids; the viscosity jump is kept, to obtain \eqref{eq:Vdot} in the next proposition. In \eqref{eq:target} the subsystem $\xi$ between $w_1$ and $w_2$ is decoupled from both, while $w_1$ and $w_2$ remain coupled to each other through \eqref{eq:t3}--\eqref{eq:t5}.

\begin{proposition}[Target decays at rate $q$ with strain-rate dissipation]\label{prop:target}
Every solution of \eqref{eq:target} satisfies
\begin{equation}\label{eq:Vdot}
\frac{d}{dt}V=-2qV-\frac{2}{\Re}\sum_j\mu_j\int_j\Big(\big|\partial^2w_j+4\pi^2k^2w_j\big|^2+16\pi^2k^2\big|\partial w_j\big|^2\Big)dy,
\end{equation}
where
\begin{equation}\label{eq:V}
V=\sum_j\rho_j\int_j\Big(|\partial w_j|^2+4\pi^2k^2|w_j|^2\Big)dy+|\xi|^2 .
\end{equation}
\end{proposition}

\begin{proof}
Multiply \eqref{eq:t1} by $\bar w_j$, integrate over each layer, sum, and take twice the real part. With $P:=\partial w_2(d)$, $W:=w_2(d)$, the left side is $-\frac{d}{dt}V_{\rm field}+2(\rho_1-\rho_2)\Rea\big(\dt{P}\bar W\big)$ after two integrations by parts, the wall terms vanishing by \eqref{eq:t2} and the interface terms combining by \eqref{eq:t3}. On the right, $\int\E^2w_j\,\bar w_j=[\partial\E w_j\,\bar w_j-\E w_j\,\partial\bar w_j]+\|\E w_j\|^2$ and $-q\rho_j\int\E w_j\,\bar w_j=q\rho_j\int(|\partial w_j|^2+4\pi^2k^2|w_j|^2)-q\rho_j[\partial w_j\bar w_j]$; the interface contributions are, by \eqref{eq:t4}--\eqref{eq:t5},
\begin{align}\label{eq:bdry}
\tfrac{2}{\Re}\,\Rea\Big\{\big[\mu_1\partial\E w_1-\mu_2\partial\E w_2\big]\bar W-\big[\mu_1\E w_1-\mu_2\E w_2\big]\bar P\Big\}
\nonumber\\
=2(\rho_1-\rho_2)\,\Rea\big((\dt{P}+qP)\bar W\big)+\tfrac{32\pi^2k^2}{\Re}(\mu_1-\mu_2)\,\Rea(P\bar W),
\end{align}
and the $q$-terms give $2qV_{\rm field}-2q(\rho_1-\rho_2)\Rea(P\bar W)$. The $\dt{P}\bar W$ and $qP\bar W$ terms cancel between the two sides, and
\begin{align}
&\sum_j\mu_j\|\E w_j\|^2+16\pi^2k^2(\mu_1-\mu_2)\Rea(P\bar W) \nonumber\\
&\qquad=\sum_j\mu_j\int_j\Big(|\partial^2w_j+4\pi^2k^2w_j|^2+16\pi^2k^2|\partial w_j|^2\Big)dy \label{eq:strain}
\end{align}
by $|\partial^2w+4\pi^2k^2w|^2=|\E w|^2+16\pi^2k^2\Rea(\partial^2w\,\bar w)$ integrated by parts on each slab with \eqref{eq:t2}--\eqref{eq:t3}. Equation \eqref{eq:t6} gives $\frac{d}{dt}|\xi|^2=-2q|\xi|^2$.
\end{proof}

\section{Spectrum of the Two-Fluid Stokes Transmission Operator}\label{sec:spectrum}

\subsection{Stokes spectrum}

\begin{proposition}[Plant is a bounded perturbation of a self-adjoint operator]\label{prop:ops}
Let $\cS_k$ be the operator of \eqref{eq:t1}--\eqref{eq:t5} at $q=0$ on $\cH_k$ with the $\eta$-row zero, and $\cB_k:=\cA_k-\cS_k$. Then $\cS_k$ is self-adjoint with compact resolvent and bounded above, and $\cB_k$ is bounded on $\cH_k$, after the change of variables $v_1\mapsto v_1-2\pi ik(1-\tfrac1M)U_1'(d)\,\eta\,\ell(y)$, with $\ell\in C^\infty[0,d]$, $\ell(0)=\ell'(0)=\ell(d)=0$, $\ell'(d)=1$. Hence $\cA_k$ and $\cA_k^*$ generate analytic semigroups on $\cH_k$.
\end{proposition}

\begin{proof}
Self-adjointness of $\cS_k$ is the symmetry of the identity in the proof of Proposition \ref{prop:target}: with the exact transmission conditions every interface term completes the strain-rate square, and the pencil term $\Re(\rho_1-\rho_2)\dt{P}$ is the boundary term of the mass form. For $\cB_k$: the interior transport term acts on $w\in\cH_k$ as $-\sum_j\rho_j\int\partial v_j\,\partial(U_j\bar w_j)-4\pi^2k^2\int U_jv_j\bar w_j$ after integration by parts, its boundary term at $y=d$ cancelling against the $2\pi ikU_I(\rho_1-\rho_2)\partial v_2(d)$ term of \eqref{eq:p8} because $U_1(d)=U_2(d)$, so it is bounded on $H^1$; the reaction term is bounded; the coefficient of $\eta$ in \eqref{eq:p8} and the $\eta$-row \eqref{eq:p9} have finite rank; the trace $v_2(d)$ in \eqref{eq:p8} is bounded on $H^1$. The term of \eqref{eq:p6} sits in the slope condition, an essential condition of the domain; the change of variables makes the domain that of $\cS_k$ and turns the term into a finite-rank bounded coupling, its time derivative being bounded through \eqref{eq:p9}.
\end{proof}

\begin{theorem}[High eigenvalues are simple and uniformly separated]\label{thm:gap}
Let $\alpha=2\pi|k|$, $\nu_j=\mu_j/\rho_j$, $\kappa=\sqrt{\nu_1/\nu_2}$, $\gamma=\rho_1/\rho_2$, and write the eigenvalues of $\cS_k$ as $\lambda=-\frac{\nu_1}{\Re}(\alpha^2+b^2)$. For $b$ large the eigenvalues are the roots of
\begin{equation}\label{eq:Phi}
\sin(bd)\cos\big(\kappa b(1-d)\big)+\sqrt{\frac{\rho_1\mu_1}{\rho_2\mu_2}}\,\cos(bd)\sin\big(\kappa b(1-d)\big)=O\Big(\frac1b\Big),
\end{equation}
For each $d\in(0,1)$ and $\rho_j,\mu_j>0$ there are $c_0>0$ and $n_1$ such that the eigenvalues $\lambda_n$ with $n\ge n_1$, indexed in increasing order of $|\lambda|$, are simple and satisfy
\begin{equation}\label{eq:gap}
d_n:=\mathrm{dist}\big(\lambda_n,\sigma(\cS_k)\setminus\{\lambda_n\}\big)\ge\frac{c_0\,n}{\Re}\qquad(n\ge n_1).
\end{equation}
\end{theorem}

\begin{proof}
Multiplicity. An eigenfunction is determined by its two free Cauchy data at $y=0$, so $\dim\ker(\cS_k-\lambda)\le2$, and self-adjointness makes algebraic multiplicity equal geometric.

Characteristic function. Put $\omega_j=\E v_j$; then $\omega_j''=\beta_j^2\omega_j$ with $\beta_j^2=\alpha^2+\Re\lambda/\nu_j=-b_j^2$, $b_2=\kappa b_1+O(1/b_1)$, and $v_j=\omega_j/(\beta_j^2-\alpha^2)+A_je^{\alpha y}+B_je^{-\alpha y}$. The wall conditions fix $A_j,B_j$. At $y=d$, to leading order in $1/b$, \eqref{eq:t4} gives $\mu_1\omega_1(d)=\mu_2\omega_2(d)$; \eqref{eq:t5} gives $\mu_1\omega_1'(d)-\mu_2\omega_2'(d)=(\gamma-1)\mu_2[\omega_2'(d)-\omega_2'(1)\cosh\alpha(1-d)]$, the inertial term being of leading order since $\Re\lambda\asymp b^2$ and $\partial v_2(d)\asymp\omega/b$; and \eqref{eq:t3} gives $\nu_1\omega_1'(0)\sinh\alpha d+\nu_2\omega_2'(1)\sinh\alpha(1-d)=0$ and $\nu_1[\omega_1'(d)-\omega_1'(0)\cosh\alpha d]=\nu_2[\omega_2'(d)-\omega_2'(1)\cosh\alpha(1-d)]$. With $\omega_1=P_1\cos b_1y+Q_1\sin b_1y$, $\omega_2=P_2\cos b_2(1-y)+Q_2\sin b_2(1-y)$, $p=b_1d$, $r=b_2(1-d)$, the $4\times4$ determinant of these conditions is
\begin{align}
&b_1b_2\mu_1^2\mu_2^2\gamma^{-2}\Big[\gamma\cosh\alpha d\,\sinh\alpha(1-d)+\cosh\alpha(1-d)\,\sinh\alpha d\Big] \nonumber\\
&\qquad\times\big(b_1\sin p\cos r+\gamma b_2\cos p\sin r\big), \label{eq:det}
\end{align}
whose bracket is positive, giving \eqref{eq:Phi} with $A:=\gamma\kappa$.

Transversality. Let $\Phi(p,r)=\sin p\cos r+A\cos p\sin r$ on $\mathbb{T}^2$ and $D=u\partial_p+v\partial_r$ with $u=d$, $v=\kappa(1-d)$, the derivative along the line $r=\kappa'p$. Where $\cos p\cos r\ne0$, $\Phi=0$ gives $\tan p=-A\tan r$ and then $Df/(\cos p\cos r)=u+Av-(Au+v)\tan p\tan r=u+Av+A(Au+v)\tan^2r>0$. Where $\cos p=0$, $\Phi=0$ forces $\cos r=0$ and $Df=-(Au+v)\sin p\sin r\ne0$; likewise where $\cos r=0$. So $\Phi$ and $Df$ have no common zero on $\mathbb{T}^2$, and by compactness $m_2:=\min_{\mathbb{T}^2}(|\Phi|+|Df|)>0$.

Gaps. By $m_2>0$, the zeros of $\Phi$ along the line are simple and separated by at least $cm_2$ in $b$, uniformly. The determinant \eqref{eq:det} is entire in $b$, and the reductions above give $b^{-m}D(b)=\Phi+R(b)$ with $|R|\le C/|b|$ uniformly on the disks of radius $cm_2/2$ about the zeros of $\Phi$ along the line and on their complement in a fixed strip, where $|\Phi|\ge m_2/2$; so Rouch\'e on each disk places exactly one exact root there, simple, and $D$ has no root outside the disks, and $\lambda=-\frac{\nu_1}{\Re}(\alpha^2+b^2)$ converts the fixed separation in $b$ into \eqref{eq:gap}.
\end{proof}

\begin{remark}
The density contrast enters the high spectrum only through $A=\sqrt{\rho_1\mu_1/(\rho_2\mu_2)}$, by the inertial interface term. At $A=1$ the characteristic function is $\sin(bd+\kappa b(1-d))$, which is the one-fluid case $\sin b$ when also $\kappa=1$. The two decoupled slab sequences, $\sin p=0$ and $\sin r=0$, collide for a dense set of $\kappa(1-d)/d$; the coupled sequence does not.
\end{remark}

\begin{lemma}[High eigenfunctions are cosines from their walls; wall coefficients grow linearly]\label{lem:doublet}
In the notation of the proof of Theorem \ref{thm:gap}, every eigenfunction of $\cS_k$ at a high eigenvalue, normalized by $\Xi:=\sum_j\int_j|\omega_j|^2=1$, satisfies
\begin{align}\label{eq:cosines}
\omega_1&=P_1\cos b_1y+O(b^{-1}),\qquad \omega_2=P_2\cos b_2(1-y)+O(b^{-1}), \nonumber\\ 1&=\tfrac d2|P_1|^2+\tfrac{1-d}2|P_2|^2+O(b^{-1}),
\end{align}
with $|P_2|\ge c>0$, and its wall traces are $\omega_2(1)=P_2$, $\omega_2'(1)=-b_2Q_2=O(1)$. For the energy-normalized eigenfunction the wall coefficients \eqref{eq:pairings} are of size $\asymp b\asymp n$.
\end{lemma}

\begin{proof}
The conditions \eqref{eq:t5} and the slope condition of \eqref{eq:t3}, at leading order in the notation of Theorem \ref{thm:gap}, are $\mu_1b_1Y=\gamma\mu_2W+(\gamma-1)\mu_2H_2b_2Q_2$ and $\nu_1(b_1Y-b_1Q_1H_1)=\nu_2(W+H_2b_2Q_2)$ with $W:=b_2(P_2\sin r-Q_2\cos r)$ and $H_j$, $S_j$ the hyperbolic cosine and sine of $\alpha d$, $\alpha(1-d)$; multiplying the second by $\rho_1$ and subtracting the first gives $\mu_1b_1H_1Q_1+\mu_2b_2H_2Q_2=0$, and the velocity condition gives $\nu_1b_1S_1Q_1-\nu_2b_2S_2Q_2=0$. In the exact transmission system these two relations hold with right-hand sides $O(b^{-1})(|P_1|+|P_2|+|Q_1|+|Q_2|)$, the corrections coming from $b_2-\kappa b_1=O(b^{-1})$, the $\omega/\kappa_j$ contributions to the traces, and the lower-order interface terms, and their matrix has determinant $-b_1b_2(\mu_1\nu_2H_1S_2+\mu_2\nu_1H_2S_1)$ with entries $O(b)$, so its inverse is $O(b^{-1})$ uniformly in the phases; hence $Q_1,Q_2=O(b^{-1})(|P_1|+|P_2|)$, which is \eqref{eq:cosines} with the normalization from the averaging of $\cos^2$; then \eqref{eq:t4} gives $\mu_1P_1\cos p=\mu_2P_2\cos r$ and \eqref{eq:t5} gives $\mu_1b_1P_1\sin p=-\gamma\mu_2b_2P_2\sin r$. At a zero of $\Phi$, $\cos r$ and $\sin r$ are not both small: where $|\cos r|\ge\frac12$ the first relation gives $|P_1|\le C|P_2|$, and where $|\cos r|<\frac12$, $\Phi=0$ forces $|\cos p|\le C|\cos r|$, hence $|\sin p|\ge\frac12$ for $b$ large, and the second relation gives $|P_1|\le C|P_2|$ with $b_2/b_1\to\kappa$; in both cases the normalization bounds $|P_2|$ below. The wall traces follow, $\omega_2'(1)=-b_2Q_2$ being $O(1)$ since $Q_2=O(b^{-1})$. Energy normalization multiplies $\omega$ by a factor $\asymp b$, since $v\sim\omega/b^2$ and $\partial v\sim\omega/b$, which gives the size of the wall coefficients.
\end{proof}

\subsection{Plant spectrum and wall coefficients}

\begin{proposition}[Root vectors of the plant form a Riesz basis with parentheses]\label{prop:riesz}
Let $n_1$ be such that $d_n>4\|\cB_k\|$ for $n\ge n_1$. For each such $n$, $\cA_k$ has exactly one eigenvalue $\lambda'_n$ in $\{|\lambda-\lambda_n|<d_n/2\}$, simple, with $|\lambda'_n-\lambda_n|\le\|\cB_k\|$, and its Riesz projection satisfies $\|P_n-P^0_n\|\le4\|\cB_k\|/d_n$. The remaining spectrum is finite and carried by $P_{\rm low}$. The family $\{P_{\rm low}\}\cup\{P_n\}_{n\ge n_1}$ is a Riesz basis with parentheses of $\cH_k$, $P_{\rm low}x+\sum_nP_nx=x$ unconditionally, and the same holds for $\cA_k^*$.
\end{proposition}

\begin{lemma}[Wall coefficients of the plant's high modes]\label{lem:wallplant}
Let $\psi_n$ be the adjoint eigenvector of $\cA_k$ at $\overline{\lambda'_n}$, normalized in $\cH_k$, and $b_n:=(\langle B_{V_c},\psi_n\rangle,\langle B_{U_c},\psi_n\rangle)$ with the pairings \eqref{eq:pairings}. Then $|(b_n)_2|\asymp n$ and $|b_n|\asymp n$ for $n\ge n_1$.
\end{lemma}

\section{Fredholm Design}\label{sec:design}

\subsection{Pre-feedback and pre-compensated plant}

The design has the structure of \cite{Kr26}: the normal wall velocity carries a static feedback of finite rank, whose only task is to make the finitely many low modes controllable from the tangential wall velocity, and the tangential wall velocity carries the Fredholm functional that assigns the whole spectrum. Both inputs are used; only one enters the infinite-dimensional construction.

The target of Theorem \ref{thm:cl} is the shifted Stokes operator $\cA^0_k-Q$ on the high modes and, on the low block, its eigenvalues made distinct: with $\lambda^0_1,\dots,\lambda^0_N$ the low Stokes eigenvalues counted with multiplicity and $0\le\epsilon_1<\dots<\epsilon_N\le1$ such that
\begin{equation}\label{eq:lowtarget}
\tilde\lambda_j:=\lambda^0_j-Q-\epsilon_j\qquad(j=1,\dots,N)\quad\text{are distinct},
\end{equation}
the target is $\tilde{\cA}_k:=\cA^0_k-Q-E$ with $E:=\sum_{j\le N}\epsilon_j\,e^0_j\otimes e^0_j$, self-adjoint and nonnegative of finite rank, so that Proposition \ref{prop:target} holds for $\tilde{\cA}_k$ with $\frac{d}{dt}V\le-2QV$.

Let $\ell\in C^\infty[d,1]$ vanish identically near $y=d$, with $\ell(1)=1$ and $\ell'(1)=0$, and $e:=(0,\ell,0)\in\cH'_k$; with the trace $\tau x:=v_2(1)$ and $P:=I-e\,\tau$, every $x\in\cH'_k$ is $x=Px+(\tau x)e$ with $Px\in\cH_k$. Vanishing near $d$ makes every interface trace of $e$ zero, so the maximal operator $\cA_{m,k}$ --- the differential operator of \eqref{eq:plant} with the tangential wall condition and all interface conditions imposed and $v_2(1)$ free --- applies to $e$; write $a_k:=\cA_{m,k}e\in\cH'_k$. Since $\ell'(1)=0$, $P$ leaves the tangential wall trace unchanged, so $B_{U_c}$ is the same operator on $\cH_k$ and on $\cH'_k$. The inner product \eqref{eq:ip} is defined on all of $\cH'_k$, so the coordinate functionals $x\mapsto\langle x,\psi_m\rangle_k$ of the Riesz basis of Proposition \ref{prop:riesz}, $\psi_m$ the adjoint root vectors, extend to $\cH'_k$. For $z$ in the resolvent set of $\cA_k$ and $c\in\mathbb{C}^2$, the response $R(z)c\in\cH'_k$ is the solution of \eqref{eq:plant} with $\partial_t$ replaced by $z$ and wall data $v_2(1)=c_1$, $\partial v_2(1)=-2\pi ikc_2$; pairing that boundary value problem with $\psi_m$ through the Green identity of Lemma \ref{lem:adjoint} gives
\begin{equation}\label{eq:response}
\big(z-\lambda'_m\big)\big\langle R(z)c,\psi_m\big\rangle_k=b_m\cdot c,\quad b_m:=\big(\langle B_{V_c},\psi_m\rangle,\langle B_{U_c},\psi_m\rangle\big),
\end{equation}
the boundary-control equation $(z-\cA_k)R(z)c=[B_{V_c}\ B_{U_c}]c$ read in the extrapolated sense. Let $A_{\rm low}:=\cA_k|_{P_{\rm low}\cH_k}$ and $B_{\rm low}=[B^V_{\rm low}\ B^U_{\rm low}]:=P_{\rm low}[B_{V_c}\ B_{U_c}]$ in the sense of \eqref{eq:response}.

\begin{lemma}[Pre-feedback through the normal velocity]\label{lem:heymann}
Let Assumption \ref{ass:d} hold. There is $K_V:P_{\rm low}\cH_k\to\mathbb{C}$ such that, with $A^{H}_{\rm low}:=A_{\rm low}+B^V_{\rm low}K_V$, every eigenvalue of $A^H_{\rm low}$ is either controllable from $B^U_{\rm low}$ or equal to one of the low target values \eqref{eq:lowtarget} with an eigenvector annihilated by $B^U_{\rm low}$, and in the latter case $A^H_{\rm low}$ is diagonalizable.
\end{lemma}

\begin{proof}
$(A_{\rm low},B_{\rm low})$ is controllable by Theorem \ref{thm:two} and Proposition \ref{prop:phi}. If $B^U_{\rm low}\ne0$, Heymann's lemma gives $F=(F_V,F_U)$ with $(A_{\rm low}+B^V_{\rm low}F_V+B^U_{\rm low}F_U,\,B^U_{\rm low})$ controllable, and since feedback through $B^U_{\rm low}$ does not change controllability from $B^U_{\rm low}$, $K_V:=F_V$ makes $(A^H_{\rm low},B^U_{\rm low})$ controllable. If $B^U_{\rm low}=0$, then $(A_{\rm low},B^V_{\rm low})$ is controllable and $K_V$ is chosen to place the eigenvalues of $A^H_{\rm low}$ at the low target values of \eqref{eq:lowtarget}, which are distinct, so $A^H_{\rm low}$ is diagonalizable; every eigenvector is then annihilated by $B^U_{\rm low}=0$.
\end{proof}

The pre-feedback is $V_c=K_Vx$, extended to $\cH'_k$ by $K_Vx:=K_VP_{\rm low}Px$, a functional of finite rank with $K_Ve=0$; the pre-fed-back plant is $\cA'_k:=\cA_k+B_{V_c}K_V$ on the hyperplane $\cH^K_k:=\{x\in\cH'_k:\ v_2(1)=K_Vx\}$, which is fixed once $K_V$ is. The map $\Xi_k:=I+e\,K_V$ is a bounded bijection of $\cH_k$ onto $\cH^K_k$ with inverse $P$: since $K_Ve=0$, $y=\Xi_kx_0$ has $\tau y=K_Vx_0=K_Vy$, so $y\in\cH^K_k$ and $Py=x_0$; conversely $\Xi_kPy=y$ for $y\in\cH^K_k$.

\begin{proposition}[Pre-fed-back plant keeps the spectral structure]\label{prop:prefed}
$\cA'_k$ generates an analytic semigroup on $\cH^K_k$ with compact resolvent, its eigenvalues $\lambda'_n$, $n\ge n_1$, are simple with $|\lambda'_n-\lambda_n|\le C$, its root vectors form a Riesz basis with parentheses of $\cH^K_k$, $\{P'_{\rm low}\}\cup\{\phi'_n\}$, with $\|P'_n-\Xi_kP^0_n\Xi_k^{-1}\|\le C/d_n$, and the same holds for its adjoint, whose high eigenvectors $\psi'_n$ satisfy $|b'_n|\asymp n$ with $b'_n:=\langle B_{U_c},\psi'_n\rangle$.
\end{proposition}

For $n\ge n_1$ the eigenvalues of $\cA'_k$ are controllable from $B_{U_c}$ by $b'_n\ne0$; the low eigenvalues are handled by Lemma \ref{lem:heymann}.

Fix a nonresonant shift $Q\ge q$: for each $k$ the set of $Q$ for which some target value --- $\lambda^0_n-Q$ for $n\ge n_1$, or $\lambda^0_j-Q-\epsilon_j$ for $j\le N$ in \eqref{eq:lowtarget} --- coincides with an eigenvalue of $\cA'_k$ is countable, being indexed by pairs of eigenvalues, so it cannot cover $[q,q+1]$. Let $R'(z)c$ be the response of the pre-fed-back plant to tangential data $\partial v_2(1)=-2\pi ikc$, a solution in $\cH^K_k$, so that $(z-\lambda'_m)\langle R'(z)c,\psi'_m\rangle_k=b'_mc$.

\begin{definition}[Design vectors]\label{def:design}
For $n\ge n_1$, with target $\tilde\lambda_n:=\lambda^0_n-Q$,
\begin{equation}\label{eq:sylvester}
c_n:=\frac{\tilde\lambda_n-\lambda'_n}{b'_n},\qquad X_n:=R'(\tilde\lambda_n)c_n\in\cH^K_k,
\end{equation}
so that $\langle X_n,\psi'_n\rangle_k=1$ and $\langle X_n,\psi'_m\rangle_k=b'_mc_n/(\tilde\lambda_n-\lambda'_m)$ for $m\ne n$. For the low block in the case $B^U_{\rm low}\ne0$: with $\varsigma:P'_{\rm low}\cH^K_k\to\mathbb{C}$ the row placing the spectrum of $A^H_{\rm low}+B^U_{\rm low}\varsigma$ at the distinct values \eqref{eq:lowtarget}, $Z$ its eigenvector matrix, and $c_j:=\varsigma Ze_j$,
\begin{equation}\label{eq:lowdesign}
X_j:=R'(\tilde\lambda_j)c_j\qquad(j=1,\dots,N),
\end{equation}
whose low coordinates are $(\tilde\lambda_j-A^H_{\rm low})^{-1}B^U_{\rm low}c_j=Ze_j$ and whose high coordinates are $b'_mc_j/(\tilde\lambda_j-\lambda'_m)$. In the case $B^U_{\rm low}=0$, $c_j:=0$ and $X_j:=\phi'_j$, the eigenvectors placed by $K_V$.
\end{definition}

\subsection{Fredholm design and closed-loop stability}

\begin{lemma}[Design vectors are quadratically close to the root vectors]\label{lem:qc}
\begin{equation}\label{eq:qc}
\sum_n\big\|X_n-\phi'_n\big\|_k^2<\infty,\qquad \sum_n|c_n|^2<\infty .
\end{equation}
\end{lemma}

\begin{lemma}[Design family is complete]\label{lem:complete}
$\{X_n\}_{n\ge n_1}$ together with the low design vectors is complete in $\cH^K_k$.
\end{lemma}

\begin{proof}
Let $\zeta\in\cH^K_k$ be orthogonal to every design vector, and for $z\in\rho(\cA'_k)$ set $G(z):=\langle R'(z)1,\zeta\rangle_k$, with the high expansion
\begin{equation}\label{eq:Gexp}
G(z)=G_{\rm low}(z)+\sum_{m\ge n_1}\frac{b'_mz_m}{z-\lambda'_m},\qquad z_m:=\langle\phi'_m,\zeta\rangle_k\in\ell^2,
\end{equation}
$G_{\rm low}$ the rational contribution of the finite block, whose poles are the eigenvalues of $A^H_{\rm low}$ with their algebraic multiplicities. By $X_n=R'(\tilde\lambda_n)c_n$ and $c_n\ne0$,
\begin{equation}\label{eq:Gzeros}
G(\tilde\lambda_n)=0,\qquad n\ge n_1 .
\end{equation}
If $B^U_{\rm low}=0$, the low design vectors are eigenvectors of $\cA'_k$ and orthogonality gives $\zeta_{\rm low}=0$, so $G_{\rm low}\equiv0$. If $B^U_{\rm low}\ne0$, then $(\tilde\lambda_j-A^H_{\rm low})Ze_j=B^U_{\rm low}c_j$ with $\tilde\lambda_j\notin\sigma(A^H_{\rm low})$ by nonresonance, so $c_j\ne0$ and \eqref{eq:Gzeros} holds at the low targets too.

Products. With $p_{\rm low}(z)=\det(zI-A^H_{\rm low})$,
\begin{equation}\label{eq:PQ}
P(z)=p_{\rm low}(z)\prod_{m\ge n_1}\Big(1-\frac{z}{\lambda'_m}\Big),\qquad
Q(z)=\prod_{j\le N}(z-\tilde\lambda_j)\prod_{m\ge n_1}\Big(1-\frac{z}{\tilde\lambda_m}\Big),
\end{equation}
both convergent since $|\lambda'_m|,|\tilde\lambda_m|\asymp m^2$. Then $H:=GP/Q$ is entire: $P$ cancels every pole of $G$, $p_{\rm low}$ those of $G_{\rm low}$ including Jordan ones, and the zeros \eqref{eq:Gzeros} cancel those of $Q$, which are simple by \eqref{eq:lowtarget} and the nonresonance.

Circles. Choose $R_N\to\infty$ midway between consecutive spectral radii. Pairing the factors, $\frac{1-z/\lambda'_m}{1-z/\tilde\lambda_m}=\frac{\tilde\lambda_m}{\lambda'_m}\cdot\frac{z-\lambda'_m}{z-\tilde\lambda_m}$; the products $\prod|\tilde\lambda_m/\lambda'_m|$ converge since $(\tilde\lambda_m-\lambda'_m)/\lambda'_m=O(m^{-2})$, and $\log|(z-\lambda'_m)/(z-\tilde\lambda_m)|=O(|\lambda'_m-\tilde\lambda_m|/|z-\tilde\lambda_m|)$ sums uniformly on $\Gamma_N=\{|z|=R_N\}$ by $|\lambda'_m-\tilde\lambda_m|\le C$ and the spacing $\gtrsim m$. Hence $\sup_N\sup_{\Gamma_N}|P/Q|<\infty$. For $G$, by Cauchy--Schwarz it suffices to bound $S_N(z)=\sum_mm^2|z-\lambda'_m|^{-2}$: for $|m-N|\le N/2$, $|z-\lambda'_m|\gtrsim N(1+|m-N|)$ and the sum is $O(1)$; for $m\le N/2$, $|z-\lambda'_m|\gtrsim N^2$ and the sum is $O(1/N)$; for $m\ge2N$, $|z-\lambda'_m|\gtrsim m^2$ and the sum is $O(1/N)$. With $G_{\rm low}$ bounded on large circles, $\sup_N\sup_{\Gamma_N}|H|<\infty$, so $H$ is a bounded entire function, hence constant.

The constant. On $z=x>0$, $\sum_mm^2(x+cm^2)^{-2}=O(x^{-1/2})$ gives $G_{\rm high}(x)=O(x^{-1/4})$, and $G_{\rm low}(x)=O(x^{-1})$, so $G(x)\to0$; $|P/Q|$ is bounded there by the same pairing, so $H\equiv0$ and $G\equiv0$.

Conclusion. The residue of $G$ at $\lambda'_m$ is $b'_mz_m$, and $b'_m\ne0$ by Proposition \ref{prop:prefed}, so $z_m=0$ for $m\ge n_1$. If $B^U_{\rm low}=0$ we already have $\zeta_{\rm low}=0$. Otherwise $G\equiv0$ reduces to $\langle(z-A^H_{\rm low})^{-1}B^U_{\rm low},\zeta_{\rm low}\rangle\equiv0$, whose expansion at infinity gives $\langle(A^H_{\rm low})^jB^U_{\rm low},\zeta_{\rm low}\rangle=0$ for $j<N$; these vectors span by controllability, so $\zeta_{\rm low}=0$. Hence $\zeta=0$.
\end{proof}

\begin{corollary}[Design family is a Riesz basis]\label{cor:basis}
$J:\phi'_n\mapsto X_n$ on the Riesz basis of $\cA'_k$, of the form $I+\mathcal{K}$ with $\mathcal{K}$ Hilbert--Schmidt by Lemma \ref{lem:qc}, is boundedly invertible, and the design family is a Riesz basis with parentheses of $\cH^K_k$.
\end{corollary}

\begin{theorem}[Closed loop is similar to the target]\label{thm:cl}
Let $U_c=F^U_kx$ with $F^U_k:\cH^K_k\to\mathbb{C}$ the bounded functional $F^U_kX_n:=c_n$ on the design basis, extended to $\cH'_k$ by $F^U_kx:=F^U_k(\Xi_kPx)$, and $V_c=K_Vx$. Then the closed loop generates an analytic semigroup on $\cH^K_k$, $\cT:\cH^K_k\to\cH_k$ defined by $\cT X_n:=e^0_n$ on the design basis, the low design vectors mapped to the eigenvectors of $\tilde{\cA}_k$ at \eqref{eq:lowtarget}, is bounded with bounded inverse, and
\begin{equation}\label{eq:intertwine}
\cT\big(\cA'_k+B_{U_c}F^U_k\big)=\tilde{\cA}_k\,\cT,
\end{equation}
where $\tilde{\cA}_k=\cA^0_k-Q-E$ is the target of \eqref{eq:lowtarget} and $\cA^0_k$ the operator of \eqref{eq:target} at $q=0$. Consequently every solution of the closed loop at wavenumber $k$ satisfies $\|x(t)\|_k\le\|\cT\|\|\cT^{-1}\|e^{-Qt}\|x(0)\|_k$.
\end{theorem}

\begin{proof}
By \eqref{eq:sylvester} and \eqref{eq:lowdesign}, every design vector solves the pre-fed-back plant at its target value with tangential data equal to its feedback value, so it is a closed-loop eigenvector; the low targets are distinct, so the closed-loop low block is diagonalizable and similar to the low block of $\tilde{\cA}_k$, and in the case $B^U_{\rm low}=0$ the low eigenvectors are those placed by $K_V$, by Lemma \ref{lem:heymann}. By Corollary \ref{cor:basis} the design family is a Riesz basis with parentheses of $\cH^K_k$, so $\cT$ is bounded with bounded inverse and $F^U_k$ is bounded by the second sum in \eqref{eq:qc} against the dual basis. Equation \eqref{eq:intertwine} holds on the design family, whose finite combinations are a core for the closed loop because they are one for $\tilde{\cA}_k$ and $\cT$ intertwines the graphs; the graph estimate on that core, the codimension-one description of $D(\cA_{{\rm cl},k})$ inside the maximal domain, and the resulting identification of the boundary realization with the operator defined on the basis are the three steps of \cite[Theorem 4]{Kr26}, with \eqref{eq:p3} in place of the single-fluid wall conditions and \eqref{eq:p5}--\eqref{eq:p9} carried unchanged. Hence \eqref{eq:intertwine} holds on $D(\cA_{{\rm cl},k})=\cT^{-1}D(\tilde{\cA}_k)$, and analyticity transports from $\tilde{\cA}_k$. The decay is Proposition \ref{prop:target} at rate $Q$ transported by $\cT$.
\end{proof}

\section{Closed Loop in Physical Variables}\label{sec:phys}

\subsection{Uncontrolled wavenumbers}

\begin{lemma}[Tail splits into a slow interface mode and fast bulk modes]\label{lem:tail}
Let $\Sigma>0$. There are $K_0$, $c$, $C$ depending on $\mu_j,\rho_j,d,\Re,G,\Sigma$ such that for $|k|\ge K_0$:
\begin{enumerate}
\item[(i)] exactly one eigenvalue $\lambda_{\rm int}(k)$ of $\cA_k$ lies in $\{|\lambda|<ck^2/\Re\}$, it is simple, and every other eigenvalue satisfies $\Rea\lambda\le-ck^2/\Re$;
\item[(ii)] $\displaystyle \Rea\,\lambda_{\rm int}(k)=-\frac{\pi\Re\,\Sigma}{\mu_1+\mu_2}\,|k|+O(1)$.
\end{enumerate}
\end{lemma}

\begin{proof}
(i) In the decomposition of Proposition \ref{prop:ops}, $\sigma(\cS_k)$ consists of the field eigenvalues, all $\le-c_1k^2/\Re$ by Theorem \ref{thm:gap} and the Poincar\'e inequality on each slab, and the simple eigenvalue $0$ of the decoupled $\eta$-row. The norm of $\cB_k$ is $O(|k|^{3/2})$: the transport term is $O(|k|)$, the reaction term $O(1)$, the coupling of $\eta$ into the fields through \eqref{eq:p8} is $O(|k|^{3/2})$ after the boundary lift, and the coupling $v_2(d)\mapsto\eta$ of \eqref{eq:p9} is $O(|k|^{3/2})$ by $|v_2(d)|\le\|v_2\|^{1/2}\|v_2\|_{H^1}^{1/2}\le C\|x\|_k/|k|^{1/2}$ against the weight $w_k^{1/2}\asymp k^2\Sigma^{1/2}$. Hence for $|k|\ge K_0$ the gap $c_1k^2/\Re$ between $0$ and the field spectrum exceeds $4\|\cB_k\|$, and the argument of Proposition \ref{prop:riesz} on the circle $\{|\lambda|=c_1k^2/(2\Re)\}$ gives exactly one eigenvalue inside, simple, with the rest within $\|\cB_k\|$ of the field spectrum.

(ii) Stretch. Let $\varepsilon=\alpha^{-1}$, $Y=\alpha(y-d)$, $L=\partial_Y^2-1$, and parametrize the branch by
\begin{equation}\label{eq:stretch}
\lambda=-i\,\mathrm{sgn}(k)\,\alpha U_I+\alpha s,\qquad s\in\overline{B(s_0,r)},\qquad s_0:=-\frac{\Re\Sigma}{2(\mu_1+\mu_2)} .
\end{equation}
Dividing \eqref{eq:p1} by $\alpha^4$ and using $U_j(d+\varepsilon Y)-U_I=\varepsilon U_j'(d)Y+O(\varepsilon^2Y^2)$, the field equation on each half-line is
\begin{equation}\label{eq:stretchfield}
\frac{\mu_j}{\Re}L^2V_j-\varepsilon\rho_j\,sLV_j+O(\varepsilon^2)=0,
\end{equation}
a regular perturbation of $L^2V_j=0$ in the exponentially weighted norm, uniformly for $s$ in the disk; the base-shear terms enter at $O(\varepsilon^2)$ because the profile is smooth and the modes decay.

Leading problem. The decaying solutions of $L^2V=0$ are $V_1=(a+b_1Y)e^Y$ on $Y<0$ and $V_2=(a+b_2Y)e^{-Y}$ on $Y>0$, $V:=a$. The kinematic condition \eqref{eq:p9} gives $\eta=\varepsilon V/s$, so the Yih term in \eqref{eq:p6} is $O(\varepsilon)$ and the slope condition reads $b_2-b_1=2V$ at leading order; \eqref{eq:p7} divided by $\alpha^2$ reads $\mu_1b_1+\mu_2b_2=(\mu_2-\mu_1)V$. Hence $b_1=-V$, $b_2=V$, and $V_1'(0)=V_2'(0)=0$.

Determinant. Write $Z_j=(V_j,V_j',LV_j,(LV_j)')^{\mathrm T}$, so that \eqref{eq:stretchfield} is $Z_j'=M_0Z_j+\varepsilon M_{1,j}(Y,s,\varepsilon)Z_j$ with $M_{1,j}$ bounded in the weighted norm, uniformly on the disk. Variation of constants gives bases of the decaying subspaces with $Z_j=Z_j^0+O(\varepsilon)e^{-c|Y|}$, hence interface trace matrices $T_j=T_j^0+O(\varepsilon)$; replacing the half-lines by $(-\alpha d,0)$ and $(0,\alpha(1-d))$ with the exact wall conditions changes them by $O(e^{-c\alpha})$, $c=\min\{d,1-d\}$. Assembling the interface conditions, eliminating $\eta=\varepsilon V/s$ and the four field amplitudes through the first three of them --- whose $4\times4$ matrix is nonsingular at $s_0$, hence uniformly invertible on the disk for large $\alpha$ --- leaves the scalar Schur complement of \eqref{eq:p8} divided by $\alpha^3$,
\begin{equation}\label{eq:Dalpha}
D_\alpha(s)=D_0(s)+O(\alpha^{-1}),\qquad D_0(s)=-2(\mu_1+\mu_2)-\frac{\Re\,\Sigma}{s},
\end{equation}
uniformly on $\overline{B(s_0,r)}$: the capillary term of \eqref{eq:p8} contributes $\Re\Sigma\,V/s$, the pencil term $O(\varepsilon^2)$ since $V_2'(0)=0$, and the gravity, $U_1'(d)$, and $8\pi^2k^2$ terms $O(\varepsilon)$ or smaller.

Rouch\'e. $D_0$ has the single zero $s_0$ in the disk, with $D_0'(s_0)=\Re\Sigma/s_0^2\ne0$. On $|s-s_0|=R/\alpha$ one has $|D_0(s)|\ge|D_0'(s_0)|R/(2\alpha)$ for large $\alpha$; choosing $R$ with $|D_0'(s_0)|R/2$ exceeding the constant in \eqref{eq:Dalpha}, Rouch\'e gives exactly one zero $s_\alpha$ with $|s_\alpha-s_0|\le R/\alpha$. By (i) that zero is $\lambda_{\rm int}$, so
\begin{equation}\label{eq:branch2}
\lambda_{\rm int}(k)=-2\pi ikU_I-\frac{\Re\,\Sigma}{2(\mu_1+\mu_2)}\,\alpha+O(1),
\end{equation}
which is (ii) with $\alpha=2\pi|k|$. With $\Sigma=0$ the elimination $\eta=\varepsilon V/s$ is singular at $s=0$ and $D_0$ has no zero, so this scaling determines nothing; the interfacial branch is then governed by gravity and by the shear coupling \cite{HB83} at a slower scale, which is not analyzed here.
\end{proof}

\begin{lemma}[Unactuated tail is uniformly exponentially stable]\label{lem:sector}
Let $\Sigma>0$ and $K\ge K_0$. The direct sum $\cA_\infty$ of the blocks $\cA_k$, $|k|\ge K$, generates an analytic semigroup on the $\ell^2$-sum of the $\cH_k$, and
\begin{equation}\label{eq:tailsg}
\|e^{t\cA_k}\|\le C\,e^{-\gamma|k|t},\qquad \|\cA_ke^{t\cA_k}\|\le\frac{C}{t}\,e^{-\gamma|k|t}\quad(t>0),\qquad \gamma=\frac{\pi\Re\,\Sigma}{2(\mu_1+\mu_2)},
\end{equation}
with $C$ independent of $k$.
\end{lemma}

\subsection{Pressure and mean flow}

\begin{proposition}[Pressure]\label{prop:pressure}
Let $k\ne0$, $\beta=2\pi k$, $\alpha=|\beta|$, let $x=(v_1,v_2,\eta)$ lie in the domain of the maximal realization $\cA_{m,k}$ of \eqref{eq:plant} --- $v_1\in H^4(0,d)$, $v_2\in H^4(d,1)$, the interface conditions \eqref{eq:p5}--\eqref{eq:p9} and the unactuated-wall conditions \eqref{eq:p3} imposed, the two traces at $y=1$ free --- and write $W:=\cA_{m,k}x=(w_1,w_2,\omega)$. With $u_j=\frac{i}{\beta}\partial v_j$, define
\begin{equation}\label{eq:P1}
i\beta p_j:=-\rho_j\frac{i}{\beta}\partial w_j-i\beta\rho_jU_j\frac{i}{\beta}\partial v_j-\rho_jU_j'v_j+\frac{\mu_j}{\Re}\E\Big(\frac{i}{\beta}\partial v_j\Big).
\end{equation}
Then $(u_1,u_2,v_1,v_2,p_1,p_2,\eta)$ satisfies both momentum equations \eqref{eq:ns} and the interface stress conditions linearized at \eqref{eq:profile}. If moreover $v_2(1)=\partial v_2(1)=0$, there are $r$, $\alpha_0$, $C$ independent of $k$ with
\begin{equation}\label{eq:P3}
\sum_j\|p_j\|_{L^2}\le C(1+\alpha)^r\big(\|W\|_k+\|x\|_k\big),\qquad \alpha\ge\alpha_0 .
\end{equation}
For every $t>0$ and every datum in $X$, the pressure of the closed-loop solution is in $L^2$ of the channel, together with the mean-flow pressure-gradient multiplier $P'(t)$.
\end{proposition}

\begin{lemma}[Mean mode]\label{lem:mean}
At $k=0$ the streamwise mean obeys
\begin{subequations}\label{eq:meanmode}
\begin{align}
\rho_j\partial_tu_j&=-P'(t)+\frac{\mu_j}{\Re}\partial^2u_j\ \ \text{on each layer},\qquad u_1(0)=0,\quad u_2(1)=U_c(0,t), \label{eq:mean}\\
0&=[u_1-u_2]_{y=d}=[\mu_1\partial u_1-\mu_2\partial u_2]_{y=d}, \label{eq:mean2}
\end{align}
\end{subequations}
with $P'(t)$ the multiplier of $\int_0^1u\dd y=0$. Its operator $\cA_0$ on $\{u:\int_0^1u=0\}$ with the inner product $\sum_j\rho_j\langle u_j,w_j\rangle$ is self-adjoint, negative, with compact resolvent and eigenvalues $\sigma_1>\sigma_2>\dots\to-\infty$, $\sigma_1\le-\gamma_0:=-\pi^2\nu_{\min}\rho_{\min}/(\rho_{\max}\Re)$, $\nu_{\min}=\min_j\nu_j$, $\rho_{\min}=\min_j\rho_j$, $\rho_{\max}=\max_j\rho_j$. An eigenvalue $\sigma$ is controllable from $U_c(0,\cdot)$ if and only if no eigenfunction at $\sigma$ has $\partial u_2(1)=0$. For every $q<q_{\rm mean}$, with $\phi_1,\dots,\phi_N$ the eigenfunctions of $\cA_0$ for the eigenvalues $\sigma_j>-q$, normalized in $\sum_j\rho_j\langle\cdot,\cdot\rangle$, the feedback
\begin{equation}\label{eq:F0}
U_c(0,t)=F_0u:=\sum_{j=1}^{N}c_j\sum_i\rho_i\langle u_i,\phi_{j,i}\rangle,
\end{equation}
with $c_1,\dots,c_N$ the solution of the placement system that assigns the eigenvalues of $\cA_0+B_0c^{\mathrm T}$ on ${\rm span}\{\phi_j\}$ to prescribed values below $-q$, makes the closed loop decay in $\sum_j\rho_j\|u_j\|^2$ at rate $q$.
\end{lemma}

\begin{proposition}[Mean-mode exception is nonempty]\label{prop:meanexc}
Let $\rho_1=\mu_1=1$, let $a_*\approx4.4934$ be the first root of $\tan a=a$ above $\pi$, and for $a\in(a_*,3\pi/2)$ set
\begin{subequations}\label{eq:fam}
\begin{align}
\delta&=\frac{1-\sec a}{2},\qquad r=\frac{\delta\pi}{\tan a-a},\qquad \rho_2=\frac1\delta,\qquad \mu_2=\frac{1}{\delta r^2}, \label{eq:family}\\
d&=\frac{ra}{\pi+ra},\qquad \beta_1=a+\frac{\pi}{r},\qquad \beta_2=r\beta_1 . \label{eq:family2}
\end{align}
\end{subequations}
Then $\sigma=-\beta_1^2/\Re$ is an eigenvalue of the mean-flow operator of Lemma \ref{lem:mean} whose eigenfunction has $\partial u_2(1)=0$, so $\sigma$ is not controllable from $U_c(0,\cdot)$; and for every $q>0$ there is $\Re$ with $\sigma>-q$.
\end{proposition}

\begin{remark}
The zero-mean constraint and the jump of $\mu_j/\rho_j$ put the mean mode outside the Volterra construction that gives the single-fluid one a closed-form kernel \cite[Lemma 15]{Kr26}, and by Proposition \ref{prop:meanexc} no construction passes $q_{\rm mean}$.
\end{remark}

For one fluid the event does not occur: with $\mu_1=\mu_2$, $\rho_1=\rho_2$, the wall and derivative conditions have the $2\times2$ minor $\cos\beta-1$, so $\beta\in2\pi\mathbb{Z}$, then $b=0$ and the flux gives $P'=0$. So $q_{\rm mean}<\infty$ occurs, and on this family $q_{\rm mean}\le\beta_1^2/\Re\to0$ as $\Re$ grows: the rate limit of Theorem \ref{thm:main} is a genuine two-fluid limitation, and it comes from a light thin low-viscosity layer at the actuated wall whose wall motion does not reach the mean profile of the fluid below.

\subsection{Proof of Theorem \ref{thm:main}}

\begin{proof}
Take $K=\max\{K_0,q/\gamma\}$ with $K_0$, $\gamma$ from Lemmas \ref{lem:tail}--\ref{lem:sector}, which is \eqref{eq:K} since $q/\gamma=2q(\mu_1+\mu_2)/(\pi\Re\Sigma)$. The controlled band is finite and for each of its wavenumbers the resonant shifts form a countable set, so their union cannot cover $[q,q+1]$ and a single $Q$ there is nonresonant for all of them.

(i) \emph{Well-posedness.} $X_F$ is the orthogonal sum of the mean space, the finitely many controlled hyperplanes $\cH^K_k$ and the tail spaces $\cH_k$, each with the norm $(2\pi|k|)^{-1}\|\cdot\|_k$. The mean block and the finitely many controlled blocks are analytic generators by Lemma \ref{lem:mean} and Theorem \ref{thm:cl}, and the tail is uniformly sectorial by Lemma \ref{lem:sector}, so their direct sum generates an analytic semigroup on $X_F$ and every datum in $X_F$ has a unique solution in $C([0,\infty);X)\cap C^1((0,\infty);X)$, with the pressure of Proposition \ref{prop:pressure} for $t>0$. The conjugate symmetry $F_{-k}\bar x=\overline{F_kx}$ follows from that of the plant and of the design, as in \cite[Lemma 6]{Kr26}, and makes the physical inputs real.

(ii) \emph{Decay.} For $k\ne0$ the modal and physical norms are related exactly by
\begin{equation}\label{eq:A1}
\|x(k)\|_{X,k}^2=\sum_j\rho_j\big(\|u_j\|^2+\|v_j\|^2\big)+\varpi^+_k|\eta(k)|^2=\frac{1}{4\pi^2k^2}\,\|x(k)\|_k^2,
\end{equation}
since $u_j=\frac{i}{2\pi k}\partial v_j$ and the interface weight in \eqref{eq:ip} is $4\pi^2k^2\varpi^+_k$. On the band, Theorem \ref{thm:cl} gives $\|x_k(t)\|_k\le C_ke^{-Qt}\|x_k(0)\|_k$ with $C_B:=\max_{0<|k|<K}C_k<\infty$; off the band, Lemma \ref{lem:sector} gives $\|x_k(t)\|_k\le C_\infty e^{-\gamma|k|t}\|x_k(0)\|_k\le C_\infty e^{-qt}\|x_k(0)\|_k$ because $K\ge q/\gamma$; at $k=0$, Lemma \ref{lem:mean} gives decay at rate $q$, its finitely many eigenvalues above $-q$ being controllable since $q<q_{\rm mean}$. With \eqref{eq:A1} and Parseval, \eqref{eq:decay} follows with $C=\max\{C_0,C_B,C_\infty\}$.

\end{proof}

\begin{figure}[t]
\centering
\includegraphics[width=\textwidth]{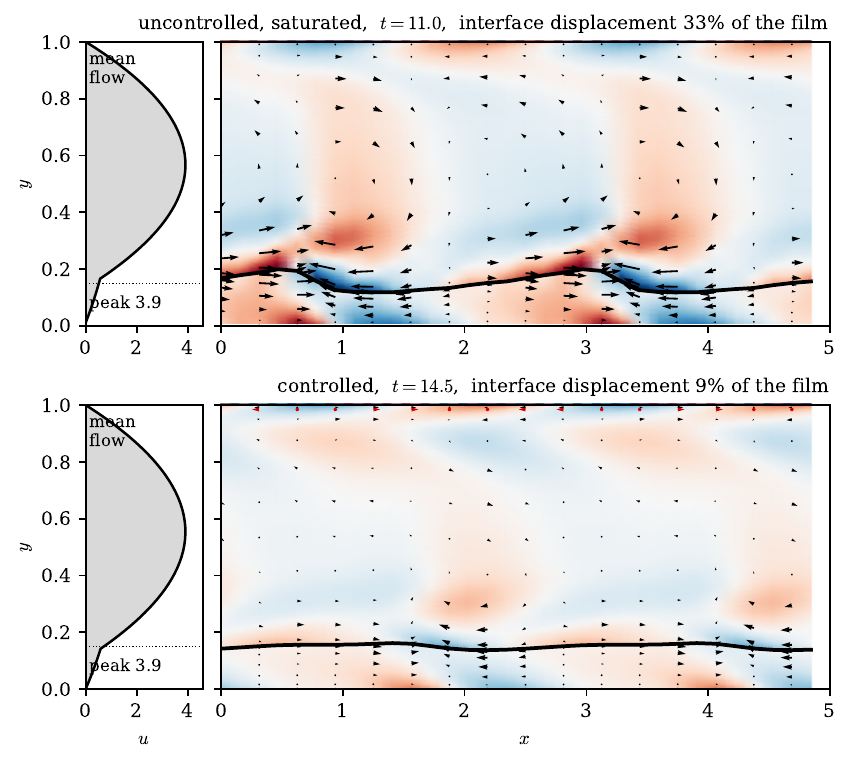}
\caption{Uncontrolled saturated state (top) and the controlled state (bottom). The instant shown in the bottom panel, $3.5$ time units after switch-on, is chosen so that the perturbation, at $9\%$ of the film thickness, is nearly extinguished but still legible; later instants show less. Left: the equilibrium profile $U_j$, peaking at $3.9$. Right: vorticity (color), velocity minus $U_j$ (arrows), and the interface (black); the wall velocity $U_c$ is drawn in red at the upper wall. Same color and arrow scales in both rows.}
\label{fig:pair}
\end{figure}

\begin{figure}[t]
\centering
\includegraphics[width=\textwidth]{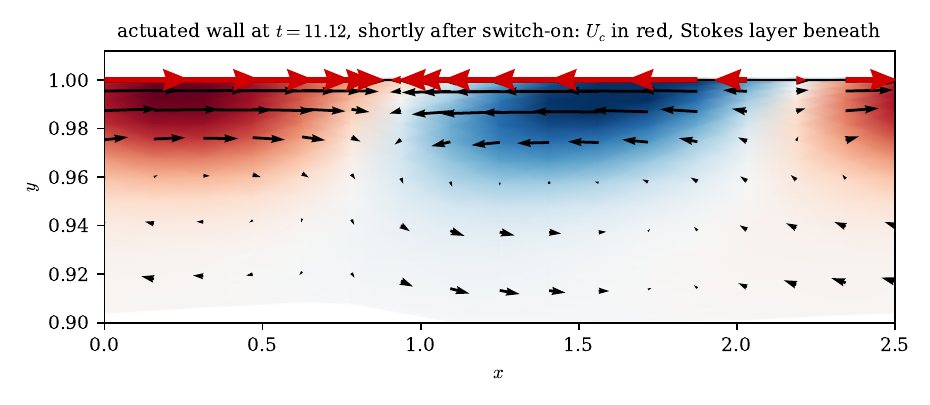}
\caption{The actuated wall, $0.9\le y\le1$, shortly after switch-on: the wall velocity $U_c$ (red) and the Stokes layer beneath it, through which the tangential actuation reaches the interface at $y\approx0.15$.}
\label{fig:wall}
\end{figure}

\section{Simulation}\label{sec:sim}

The nonlinear two-fluid system \eqref{eq:ns}--\eqref{eq:ifc} is simulated with $\Re=60$, $M=0.2$, $d=0.15$, $\Sigma=0.05$, matched densities, $G=0$, and period $L=5$: a thin film five times more viscous than the fluid above it, lying against the unactuated wall. Two wavenumbers are unstable, $k=1/5$ and $k=2/5$, with growth rates $0.26$ and $0.56$. Each layer is mapped to a fixed reference slab so that the interface is a coordinate line, with $32$ Fourier modes in $x$, $24$ and $64$ Chebyshev points in the two layers, the interface conditions \eqref{eq:ifc} imposed in exact form at every step by a Newton solve, and the time step $2\times10^{-4}$.

The controller is tangential only, built by the construction of Section \ref{sec:design} with $V_c\equiv0$ rather than by Theorem \ref{thm:main}, the assigned modes having nonzero wall coefficients at these parameters so that $K_V$ is not needed. It places the leading interfacial eigenvalue of $k\in\{1/5,2/5,3/5\}$ at real part $-1$.

Left alone, a perturbation of the flat interface of amplitude $10^{-3}$ grows at the linear rate, steepens, and saturates by $t\approx10.5$ into a travelling wave whose crests and troughs displace the interface by a third of the film thickness (Figure \ref{fig:pair}, top). The controller is switched on at $t=11$ from this saturated state, with no ramp. The interface amplitude falls to $9\%$ of the film thickness by $t=14.5$ (Figure \ref{fig:pair}, bottom) and to $5\%$ by $t=15.6$; the peak wall velocity is $1.3$ at switch-on, against a base flow peaking at $3.9$, and decreases with the interface to $0.05$. The residual at $5\%$ lies in the harmonics outside the controlled band, which decay at their own rate. Figure \ref{fig:wall} shows the actuated wall shortly after switch-on: the wall velocity and the Stokes layer through which it acts on the fluid.

\section{Conclusions}

The state assigned a decay rate here is the position of the interface, and the actuator reaches it only through the fluid that lies between them. Two constraints follow. The actuation must reach across that fluid, and the target may alter only some of the conditions holding at the interface. The second constraint is the familiar one for plants assembled from subsystems joined at an internal boundary, as in the sandwiched systems of \cite{WK20,WK21}: the transformation may delete the coupling that destabilizes, but not the junction itself. Section \ref{sec:target} carries out that division for an interface at which the junction is the transmission of stress between two fluids.

Interconnection usually makes a spectrum harder to work with. Here it does the opposite: the transmission conditions separate sequences that are dense in one another while the layers are uncoupled. Whether this occurs at other junctions and for other operators is not settled by one case.

The instability controlled here is not turbulence. It occurs at arbitrarily small Reynolds number and comes from a contrast of material properties rather than from inertia. Suppressing it is therefore a different problem from delaying transition: the interface is held flat and the bulk flow remains laminar.

Three questions are left open: whether the interface-position assumption can be removed; what governs the interfacial branch when there is no surface tension, where the capillary scaling of Lemma \ref{lem:tail} degenerates; and whether a design of this kind carries over to the diffuse-interface description \cite{Mun21}, in which the state is a field and the domain does not move.

\appendix

\section{Proofs of Supporting Results}

\subsection{Spectrum and wall coefficients}

\begin{proof}[Proof of Proposition \ref{prop:riesz}]
As in \cite[Prop.\ 4 and Lemma 5]{Kr26}: on $\Gamma_n=\{|\lambda-\lambda_n|=d_n/2\}$, $\|\cB_k(\lambda-\cS_k)^{-1}\|\le2\|\cB_k\|/d_n<\frac12$ by self-adjointness of $\cS_k$ and Theorem \ref{thm:gap}, the Neumann series gives the resolvent bound and the projection estimate, and $\sum_n\|P_n-P^0_n\|^2\le16\|\cB_k\|^2\sum_nd_n^{-2}<\infty$ by \eqref{eq:gap}; the Bari--Markus lemma of \cite{Kr26} then applies. The eigenvalue localization is the inclusion $\sigma(\cA_k)\subset\{z:\mathrm{dist}(z,\sigma(\cS_k))\le\|\cB_k\|\}$ read inside each disk.
\end{proof}

\begin{proof}[Proof of Lemma \ref{lem:wallplant}]
Let $\psi_n=(\zeta_1,\zeta_2,\chi)$ be the adjoint eigenvector at $\overline{\lambda'_n}$, normalized in $\cH_k$, and $\omega_j:=\E\zeta_j$, $b_j^2:=-\alpha^2-\Re\,\overline{\lambda'_n}\rho_j/\mu_j$, so $b_j\asymp n$ and $b_2/b_1\to\kappa$ by Proposition \ref{prop:riesz}.

Field equation. On each layer, \eqref{eq:a1} reads $\omega_j''+b_j^2\omega_j=g_j$ with $g_j:=-\frac{2\pi ik\rho_j\Re}{\mu_j}\big(U_j\omega_j+2U_j'\zeta_j'\big)$, since $\E(U\zeta)-U''\zeta=U\E\zeta+2U'\zeta'$. Variation of constants from the wall of layer $j$ gives $\omega_2=P_2\cos b_2(1-y)+Q_2\sin b_2(1-y)+r_2$ with $r_2(1)=r_2'(1)=0$ and $\|r_2\|_\infty\le\|g_2\|_{L^1}/b_2$, and likewise on $(0,d)$ from $y=0$; as $\zeta_j'=O(\|\omega_j\|/b_j)$ from $\zeta_j=\omega_j/(\beta_j^2-\alpha^2)+$ evanescent terms, $\|g_j\|\le C\|\omega_j\|$ and the remainders are $O(\|\omega\|/b)$ in the sup norm. The wall traces are exact: $\E\zeta_2(1)=P_2$, $\partial\E\zeta_2(1)=-b_2Q_2$.

Interface scalar. In \eqref{eq:a6} the left side is $w_k\Re(\overline{\lambda'_n}-2\pi ikU_I)\chi$ with $|\lambda'_n|\asymp n^2$, and the right side is bounded by $C(|\E\zeta_1(d)|+n^2|\zeta(d)|)$. In the normalization, $\|\partial\zeta_j\|\le1$ gives $|\zeta(d)|\le C/n$ and $|\E\zeta_1(d)|\le Cn$ by the representation above, so $|\chi|\le C/n$, $w_k|\chi|^2=O(n^{-2})$, and the term $-\Re w_k\chi$ in \eqref{eq:a5} is $O(1/n)$ against traces of size $n^2$: the interface scalar is lower order, and the normalization is carried by the fields, $\sum_j\rho_j\|\omega_j\|^2/b_j^2=1+O(n^{-1})$, so $\|\omega\|\asymp n$.

Amplitudes. With $\chi$ and the remainders lower order, the conditions \eqref{eq:a3}--\eqref{eq:a5} at $y=d$ are those of the proof of Lemma \ref{lem:doublet} up to relative $O(1/n)$: the $U_1'(d)\zeta(d)$ term of \eqref{eq:a5} is $O(1/n)$ against $\partial\E\zeta$ of size $n^2$, and the transport coefficient $\overline{\lambda'_n}-2\pi ikU_I$ differs from $\overline{\lambda'_n}$ by $O(1)$. Hence $Q_1,Q_2=O(\|\omega\|/n)$ and, by the two-case argument of Lemma \ref{lem:doublet}, $|P_2|\ge c\|\omega\|$, so $|P_2|\asymp n$, while $|b_2Q_2|\le Cn$.

Conclusion. $|b_n|\ge2\pi|k|\frac{\mu_2}{\Re}|\E\zeta_2(1)|=2\pi|k|\frac{\mu_2}{\Re}|P_2|\ge cn$, and $|b_n|\le\frac{\mu_2}{\Re}(|b_2Q_2|+2\pi|k||P_2|)\le Cn$.
\end{proof}

\subsection{Pre-compensated plant and design family}

\begin{proof}[Proof of Proposition \ref{prop:prefed}]
Conjugation. For $x\in D(\cA_k)$, $\Xi_kx\in D(\cA'_k)$ and $\cA'_k\Xi_kx=\cA_kx+a_k\,K_Vx$, so
\begin{equation}\label{eq:conj}
\cA''_k:=\Xi_k^{-1}\cA'_k\Xi_k=\cA_k+\Theta_k,\qquad \Theta_k:=(Pa_k)\otimes K_V,
\end{equation}
of rank one and bounded, $\|\Theta_k\|\le\|Pa_k\|_k\|K_V\|$, a fortiori $\cA_k$-bounded with relative bound zero. Hence $\cA''_k$, and with it $\cA'_k$, generates an analytic semigroup with compact resolvent.

Spectrum and basis. $\cA''_k=\cS_k+(\cB_k+\Theta_k)$ with the bracket bounded, so the argument of Proposition \ref{prop:riesz} applies with $C_k:=\|\cB_k\|+\|\Theta_k\|$ in place of $\|\cB_k\|$: on $\Gamma_n$, $\|(\cB_k+\Theta_k)(\cS_k-z)^{-1}\|\le2C_k/d_n<\frac12$ once $d_n>4C_k$, giving one simple eigenvalue $\lambda'_n$ per disk with $|\lambda'_n-\lambda_n|\le C_k$ and $\|P''_n-P^0_n\|\le4C_k/d_n$, square-summable by Theorem \ref{thm:gap}, hence a Riesz basis with parentheses of $\cH_k$; $\Xi_k$ carries it to $\cH^K_k$, with $\|P'_n-\Xi_kP^0_n\Xi_k^{-1}\|\le\|\Xi_k\|\|\Xi_k^{-1}\|\,4C_k/d_n$. The same holds for the adjoint.

Wall coefficients. Since $\Theta_k$ is bounded, $D\big((\cA''_k)^*\big)=D(\cA_k^*)$ and $(\cA''_k)^*=\cA_k^*+\Theta_k^*$: the adjoint conditions are those of Lemma \ref{lem:adjoint} unchanged, in particular $\zeta_2(1)=\partial\zeta_2(1)=0$, and since $\ell'(1)=0$ the tangential pairing is \eqref{eq:pairings} unchanged, $b'_n=-2\pi ik\frac{\mu_2}{\Re}\overline{\E\zeta''_{2,n}(1)}$. The normalized adjoint eigenvector satisfies $(\cA_k^*-\overline{\lambda'_n})\psi''_n=-\Theta_k^*\psi''_n$, and $\Theta_k^*\psi=k_K\langle\psi,Pa_k\rangle_k$ with $k_K$ the representative of $K_V$, which lies in the fixed finite-dimensional range of $P^*_{\rm low}$ and is smooth; so $\|\Theta_k^*\psi''_n\|_k\le\|\Theta_k\|$ and the field equation of Lemma \ref{lem:wallplant} acquires a forcing $r_{j,n}$ from a fixed finite-dimensional family with $\|r_{j,n}\|_{L^1}+\|r_{j,n}\|_{L^2}\le C$. The variation-of-constants remainder then gains $C/n$, the $2\times2$ system for $Q_1,Q_2$ an $O(1/n)$ inhomogeneity, and the two leading interface relations $o(\|\omega_n\|)$ and $o(n\|\omega_n\|)$ respectively. Since $\|P''_n-P^0_n\|=O(1/n)$ keeps the normalization in the fields, $\|\omega_n\|\asymp n$, and the two-case argument of Lemma \ref{lem:doublet} gives $|P_{2,n}|\ge c\|\omega_n\|$, hence $|P_{2,n}|\asymp n$. Launching the variation of constants at $y=1$ makes the remainder and its derivative vanish there, so $\E\zeta''_{2,n}(1)=P_{2,n}$ exactly and $|b'_n|\asymp n$.
\end{proof}

\begin{proof}[Proof of Lemma \ref{lem:qc}]
$|c_n|=|\tilde\lambda_n-\lambda'_n|/|b'_n|\le C/n$ for $n\ge n_1$ by Proposition \ref{prop:prefed}, so the second sum converges. For $m\ne n$, $|\langle X_n,\psi'_m\rangle_k|=|b'_m||c_n|/|\tilde\lambda_n-\lambda'_m|\le C(m/n)\Re/|m^2-n^2|$ by Theorem \ref{thm:gap}, Proposition \ref{prop:prefed}, and the nonresonance of $Q$; the low components are $O(n^{-2})$ by the same bound. Hence
\begin{equation}\label{eq:tails}
\sum_{m\ne n}|\langle X_n,\psi'_m\rangle_k|^2\le C\Big(\sum_{m\le n/2}\frac{m^2}{n^6}+\frac{1}{n^2}\sum_{j\ge1}\frac{1}{j^2}+\sum_{m\ge2n}\frac{1}{n^2m^2}\Big)\le\frac{C'}{n^2},
\end{equation}
where the middle sum runs over $m=n\pm j$, and since the parentheses are those of a Riesz basis, $\|X_n-\phi'_n\|_k^2\le C''n^{-2}$; the first sum converges.
\end{proof}

\begin{proof}[Proof of Corollary \ref{cor:basis}]
$J$ is Fredholm of index zero, so its range is closed; Lemma \ref{lem:complete} makes it dense, hence onto, and index zero gives $\ker J=\{0\}$.
\end{proof}

\subsection{Tail, pressure, and mean flow}

\begin{proof}[Proof of Lemma \ref{lem:sector}]
Notation. $P^0_k$ is the Stokes projection onto the $\eta$-mode, $Q^0_k=I-P^0_k$, and $\cS_k=0\oplus S_{f,k}$ with $S_{f,k}\le-\delta_\alpha$, $\delta_\alpha:=a_0\alpha^2/\Re$, $a_0$ independent of $k$ by Theorem \ref{thm:gap} and the Poincar\'e inequality; $\|\cB_k\|\le M\alpha^{3/2}$ by the proof of Lemma \ref{lem:tail}(i).

Projections. Enlarge $K_0$ so that $\delta_\alpha\ge4\gamma|k|$ for $|k|\ge K_0$, possible since $\delta_\alpha=a_0\alpha^2/\Re$ and $\gamma|k|$ is linear. For $\alpha\ge(4M\Re/a_0)^2$ one has $2\|\cB_k\|/\delta_\alpha\le\frac12$, so on $\Gamma_\alpha=\{|z|=\delta_\alpha/2\}$, $\|(z-\cS_k)^{-1}\|\le2/\delta_\alpha$ and $\|(z-\cA_k)^{-1}\|\le4/\delta_\alpha$, whence with $P_k$ the Riesz projection of $\lambda_{\rm int}(k)$,
\begin{equation}\label{eq:projest}
P_k-P^0_k=\frac{1}{2\pi i}\int_{\Gamma_\alpha}(z-\cS_k)^{-1}\cB_k(z-\cA_k)^{-1}dz,\qquad \|P_k-P^0_k\|\le\frac{4M\Re}{a_0\sqrt\alpha} .
\end{equation}
Since $\cS_k(z-\cS_k)^{-1}=z(z-\cS_k)^{-1}-I$ is bounded by $2$ on $\Gamma_\alpha$, the same integral gives
\begin{equation}\label{eq:graphest}
\big\|\cS_k\big(P_k-P^0_k\big)\big\|\le 4M\alpha^{3/2}.
\end{equation}

Straightening. By \eqref{eq:projest}, for large $\alpha$ the Kato intertwiner $W_k$ with $W_kP^0_k=P_kW_k$ exists, $\|W_k-I\|+\|W_k^{-1}-I\|\le C_1/\sqrt\alpha$, and, being a function of $P_k-P^0_k$, it satisfies $\|\cS_k(W_k-I)\|\le C_2\alpha^{3/2}$ by \eqref{eq:graphest}. Then $\widehat{\cA}_k:=W_k^{-1}\cA_kW_k$ commutes with $P^0_k$ and is block diagonal, $\widehat{\cA}_k=\lambda_{\rm int}(k)\oplus A_{f,k}$: the $O(\alpha^2)$ interaction that produces the capillary branch is absorbed into the scalar block by Lemma \ref{lem:tail}(ii) and is not estimated again.

Fast block. For $\Rea z\ge-\delta_\alpha/2$, $\|(z-S_{f,k})^{-1}\|\le2/\delta_\alpha$ and $\|S_{f,k}(z-S_{f,k})^{-1}\|\le C$. From $\widehat{\cA}_k-\cS_k=W_k^{-1}[\cS_k,W_k]+W_k^{-1}\cB_kW_k$ and the three estimates above,
\begin{equation}\label{eq:fastpert}
\big\|(A_{f,k}-S_{f,k})(z-S_{f,k})^{-1}\big\|\le C\Big(\|W_k-I\|+\frac{\alpha^{3/2}}{\delta_\alpha}\Big)\le\frac{C_*}{\sqrt\alpha},\qquad C_*=C\Big(C_1+\frac{M\Re}{a_0}\Big),
\end{equation}
so for $\alpha\ge4C_*^2$ the Neumann series gives $\|(z-A_{f,k})^{-1}\|\le2\|(z-S_{f,k})^{-1}\|$ there, and the same on the sector $|\arg(z+\delta_\alpha/2)|\le\pi-\theta$, with constants independent of $k$. Hence $\|e^{tA_{f,k}}\|\le Ce^{-\delta_\alpha t/2}$ and $\|A_{f,k}e^{tA_{f,k}}\|\le Ct^{-1}e^{-\delta_\alpha t/4}$.

Slow block and assembly. By Lemma \ref{lem:tail}(ii) the rate is $\Rea\lambda_{\rm int}(k)=-2\gamma|k|+O(1)$, so after enlarging $K_0$, $\Rea\lambda_{\rm int}(k)\le-\tfrac32\gamma|k|$, and with $|\lambda_{\rm int}(k)|\le C\alpha$, $|\lambda_{\rm int}|e^{t\Rea\lambda_{\rm int}}\le C|k|e^{-3\gamma|k|t/2}\le Ct^{-1}e^{-\gamma|k|t}$. Since $\delta_\alpha/4\ge\gamma|k|$ and $W_k$, $W_k^{-1}$ are uniformly bounded, transporting the two blocks back gives \eqref{eq:tailsg}. The direct sum inherits the sector estimate with the same constants.
\end{proof}

\begin{proof}[Proof of Proposition \ref{prop:pressure}]
Equations. \eqref{eq:P1} is the Fourier transform of \eqref{eq:mx} linearized at \eqref{eq:profile}. For the normal component, let $R_j:=\rho_j(\partial_tv_j+i\beta U_jv_j)+\partial p_j-\frac{\mu_j}{\Re}\E v_j$; differentiating \eqref{eq:P1} in $y$, subtracting $i\beta R_j$, and using $i\beta u_j+\partial v_j=0$ leaves exactly \eqref{eq:p1}, so $i\beta R_j=0$ and $R_j=0$ for $k\ne0$. At the interface, \eqref{eq:p7} is the tangential-stress condition, which carries no pressure, and substituting \eqref{eq:P1} into the primitive normal-stress jump returns \eqref{eq:p8}; no interface condition is lost in the reduction to \eqref{eq:plant}.

Parameter-elliptic estimate. Let $v_2(1)=\partial v_2(1)=0$. From \eqref{eq:p1}, $\E^2v_j=r_j$ with $r_j:=\frac{\Re}{\mu_j}\big(\rho_j\E w_j+i\beta\rho_jU_j\E v_j-i\beta\rho_jU_j''v_j\big)$ and $\|r\|_{H^{-1}}\le C\alpha(\|W\|_k+\|x\|_k)$. Scale $Y=\alpha y$: then $\E^2=\alpha^4(\partial_Y^2-1)^2$, the walls sit at distance $\alpha d$ and $\alpha(1-d)$ from the interface, and the conditions \eqref{eq:p3}, \eqref{eq:p5}--\eqref{eq:p8}, each divided by the power of $\alpha$ that normalizes its leading term, become $\alpha$-independent with $O(\alpha^{-1})$ corrections and with $\eta$ and $\omega$ appearing as data. In the limit the problem decouples into three problems with $O(e^{-c\alpha})$ coupling, $c=\min\{d,1-d\}$.

At a clamped wall, $(\partial_Y^2-1)^2V=F$ on a half-line with $V=\partial_YV=0$ at the end and $V$ decaying: the homogeneous decaying solutions are $(a+bY)e^{-Y}$, and the two conditions give $a=b=0$, so the problem is uniquely solvable. At the interface, with $V_1=(a_1+b_1Y)e^{Y}$ on $Y<0$ and $V_2=(a_2+b_2Y)e^{-Y}$ on $Y>0$, the four conditions read
\begin{equation}\label{eq:4x4}
a_1-a_2=d_1,\quad a_1+a_2+b_1-b_2=d_2,\quad 2\mu_1b_1+2\mu_2b_2=d_3,\quad 2\mu_1b_1-2\mu_2b_2=d_4,
\end{equation}
using $\E V_1(0)=2b_1$, $\E V_2(0)=-2b_2$, $(\E V_1)'(0)=2b_1$, $(\E V_2)'(0)=2b_2$. The last two give $b_1=(d_3+d_4)/(4\mu_1)$, $b_2=(d_3-d_4)/(4\mu_2)$, and the first two then give $a_1,a_2$: after the column permutation to the order $(a_1,a_2,b_1,b_2)$ the matrix of \eqref{eq:4x4} is block triangular with determinant $-16\mu_1\mu_2$, so in the order $(a_1,b_1,a_2,b_2)$ its determinant is $16\mu_1\mu_2\ne0$. Nonvanishing of this determinant and of the two clamped-wall ones is the Lopatinskii--Shapiro condition for the limiting transmission problem, so the two-slab problem is parameter-elliptic in $\alpha$ and its scaled solution operator is bounded uniformly for $\alpha\ge\alpha_0$, and undoing the scaling gives
\begin{equation}\label{eq:parell}
\sum_j\|v_j\|_{H^3}\le C(1+\alpha)^r\big(\|W\|_k+\|x\|_k\big),
\end{equation}
whence \eqref{eq:P3} by \eqref{eq:P1}.

Closed loop. For $|k|\ge K$ the state has $v_2(1)=\partial v_2(1)=0$ and $W=\cA_kx$, so \eqref{eq:P3} applies with the uniform constant; for the finitely many $0<|k|<K$ the traces are $K_Vx$ and $-2\pi ikF^U_kx$, and for such a fixed $k$ ordinary elliptic regularity for the fourth-order two-slab problem with those bounded feedback boundary conditions gives $\sum_j\|v_j\|_{H^3}\le C_k(\|\cA_{{\rm cl},k}x\|_k+\|x\|_k)$ on $D(\cA_{{\rm cl},k})$, hence $\sum_j\|p_j\|\le C_k(\|\cA_{{\rm cl},k}x\|_k+\|x\|_k)$ by \eqref{eq:P1}; no uniformity in $k$ is needed on a finite set, and the analyticity of Theorem \ref{thm:cl} then supplies $\|\cA_{{\rm cl},k}x(t)\|_k\le C_kt^{-1}\|x(0)\|_k$ for $t>0$.

Summability. For $|k|\ge K$, Lemma \ref{lem:sector} at $t/2$ gives $\|W\|_k+\|x_k(t)\|_k\le Ct^{-1}e^{-\gamma|k|t/2}\|x_k(0)\|_k$, so $\|p_k(t)\|\le C_t(1+\alpha)^{r}e^{-\gamma|k|t/2}\|x_k(0)\|_k$. With $\|x_k(0)\|_k^2=4\pi^2k^2E_k(0)$, $E_k$ the modal contribution to $\|\cdot\|_X^2$, the factor $(1+\alpha)^{2r+2}e^{-\gamma|k|t}$ is bounded uniformly in $k$ for fixed $t>0$, so $\sum_{|k|\ge K}\|p_k(t)\|^2\le C_t\|x(0)\|_X^2$; the band contributes finitely many terms. At $k=0$ the pressure is not a Fourier mode but the scalar multiplier $P'(t)$ of the zero-flux constraint.
\end{proof}

\begin{proof}[Proof of Lemma \ref{lem:mean}]
For $u,w$ in the domain, $\sum_j\rho_j\langle\cA_0u,w\rangle=\sum_j\frac{\mu_j}{\Re}\int_ju_j''\bar w_j-P'(u)\int_0^1\bar w=-\sum_j\frac{\mu_j}{\Re}\int_ju_j'\bar w_j'$, the interface terms cancelling by the transmission conditions, the wall terms by $w(0)=w(1)=0$, and the multiplier term by $\int w=0$; this is symmetric and negative. With $\nu_{\min}=\min_j\nu_j$, $\rho_{\min}=\min_j\rho_j$, $\rho_{\max}=\max_j\rho_j$, the Poincar\'e inequality on $(0,1)$ gives $\frac1\Re\sum_j\mu_j\|u_j'\|^2\ge\frac{\nu_{\min}\rho_{\min}}{\Re}\|u'\|^2\ge\frac{\pi^2\nu_{\min}\rho_{\min}}{\Re\,\rho_{\max}}\sum_j\rho_j\|u_j\|^2$, so $\sigma_1\le-\gamma_0$. With $u_2(1)=U\ne0$ the same identity acquires the boundary term $-\frac{\mu_2}{\Re}U\,\overline{\partial w_2(1)}$, so the input functional on the adjoint is $w\mapsto-\frac{\mu_2}{\Re}\overline{\partial w_2(1)}$, and Fattorini's criterion is the stated one. For $q<q_{\rm mean}$ the finitely many $\sigma_j>-q$ are controllable by \eqref{eq:qmean}, and finite-dimensional pole placement on their span gives $F_0$.
\end{proof}

\begin{proof}[Proof of Proposition \ref{prop:meanexc}]
On $(a_*,3\pi/2)$, $\cos a<0$ and $\tan a>a$, so $\delta,r>0$ and all parameters are admissible. With $P'=\beta_1^2/\Re$, the functions $u_1=1-\cos\beta_1y-\tan a\,\sin\beta_1y$ and $u_2=\delta[1-\cos\beta_2(1-y)]$ satisfy $\frac{\mu_j}{\Re}u_j''-\sigma\rho_ju_j=P'$ on each layer, since $\beta_2^2/\beta_1^2=\rho_2\mu_1/(\rho_1\mu_2)=r^2$ and $\delta=\beta_1^2/(\mu_2\beta_2^2)$; $u_1(0)=0$ and $u_2(1)=\partial u_2(1)=0$; at $y=d$, $\beta_1d=a$ and $\beta_2(1-d)=\pi$ give $u_2(d)=2\delta$, $\partial u_2(d)=0$, $\partial u_1(d)=\beta_1(\sin a-\tan a\cos a)=0$, and $u_1(d)=1-\cos a-\tan a\sin a=1-\sec a=2\delta$; and $\int_0^du_1=(a-\tan a)/\beta_1$, $\int_d^1u_2=\delta\pi/\beta_2$, whose sum vanishes by the definition of $r$. Hence $u$ is an eigenfunction at $\sigma$ with zero wall shear. The last claim follows from $\sigma=-\beta_1^2/\Re$ with $\beta_1$ independent of $\Re$.
\end{proof}

\subsection{Adjoint of plant}

\begin{proof}[Proof of Lemma \ref{lem:adjoint}]
For $x$ in the domain of $\cA_k$ and $\psi$ in the domain of $\cA_k^*$, $\langle\cA_kx,\psi\rangle_k-\langle x,\cA_k^*\psi\rangle_k$ is the sum over the two slabs of the Green identity of the fourth-order operator $\frac{\mu_j}{\Re}\E^2-2\pi ik\rho_jU_j\E+2\pi ik\rho_jU_j''$ against its formal $L^2$-adjoint, the left side of \eqref{eq:a1}, the mass form $\rho_j\E$ contributing $-\lambda\rho_j[\partial v_j\bar\zeta_j-v_j\partial\bar\zeta_j]$ at the endpoints, and the $\eta$-row pairing $w_k(\lambda\eta+2\pi ikU_I\eta-v_2(d))\bar\chi$. The boundary form at each endpoint of layer $j$ is
\begin{align}
&\frac{\mu_j}{\Re}\Big[\partial\E v\,\bar\zeta-\E v\,\partial\bar\zeta+\partial v\,\overline{\E\zeta}-v\,\overline{\partial\E\zeta}\Big] \nonumber\\
&\qquad-\big(2\pi ik\rho_jU_j+\lambda\rho_j\big)\big[\partial v\,\bar\zeta-v\,\partial\bar\zeta\big]+2\pi ik\rho_jU_j'\,v\bar\zeta, \label{eq:green}
\end{align}
and the reaction term contributes nothing since $U_j''$ is real. At $y=0$ and $y=1$ the plant conditions \eqref{eq:p3}--\eqref{eq:p4b} at zero input leave the terms in $\zeta_j$, $\partial\zeta_j$, which \eqref{eq:a2} removes. At $y=d$, substituting the plant conditions \eqref{eq:p5}--\eqref{eq:p8} for the layer-1 traces leaves a linear form in the five free quantities $v_2(d)$, $\partial v_2(d)$, $\E v_2(d)$, $\partial\E v_2(d)$, $\eta$; requiring it to vanish identically gives five conditions on the traces of $\zeta$ and on $\chi$, which after conjugation are \eqref{eq:adj}: the coefficients of $\partial\E v_2(d)$ and $\E v_2(d)$ give \eqref{eq:a3}, that of $\partial v_2(d)$ gives \eqref{eq:a4}, that of $v_2(d)$ gives \eqref{eq:a5}, and that of $\eta$ gives \eqref{eq:a6}. The algebra was carried out symbolically.
\end{proof}

\section*{Acknowledgment}
The author's problems, ideas, and results were developed with the assistance of Claude and ChatGPT in final theorem formulation, proofs, simulations, and drafting throughout the paper, under the author's correction and complete verification.

\end{document}